\documentclass[a4paper]{amsart}
\usepackage{amscd}
\usepackage{amsmath}
\usepackage{amssymb}
\usepackage{mathrsfs}
\usepackage{amsthm}
\usepackage{bbm} 
\usepackage{stmaryrd}

\usepackage{lipsum}

\makeatletter
\renewcommand\subsection{\@startsection{subsection}{2}%
  \z@{-.5\linespacing\@plus-.7\linespacing}{.5\linespacing}%
  {\normalfont\scshape}}
\renewcommand\subsubsection{\@startsection{subsubsection}{3}%
  \z@{.5\linespacing\@plus.7\linespacing}{-.5em}%
  {\normalfont\scshape}}
\makeatother

\usepackage[T1]{fontenc}

\numberwithin{equation}{section} \swapnumbers

\newtheorem{satz}{Satz}[section]

\newtheorem{theorem}[satz]{Theorem}
\newtheorem{proposition}[satz]{Proposition}
\newtheorem{corollary}[satz]{Corollary}
\newtheorem{lemma}[satz]{Lemma}

\newtheorem{definition}[satz]{Definition}

\newtheorem{remark}[satz]{Remark}

\newtheorem{example}[satz]{Example}

\newcommand{\bbr}{\mathbb{R}}
\newcommand{\bbe}{\mathbb{E}}
\newcommand{\bbn}{\mathbb{N}}
\newcommand{\bbp}{\mathbb{P}}
\newcommand{\bbq}{\mathbb{Q}}

\newcommand{\bbs}{\mathbb{S}}

\newcommand{\cala}{\mathscr{A}}
\newcommand{\calb}{\mathscr{B}}

\newcommand{\cale}{\mathscr{E}}
\newcommand{\calf}{\mathscr{F}}

\newcommand{\calh}{\mathscr{H}}

\newcommand{\call}{\mathscr{L}}
\newcommand{\calm}{\mathscr{M}}

\newcommand{\calo}{\mathscr{O}}
\newcommand{\calp}{\mathscr{P}}
\newcommand{\cals}{\mathscr{S}}

\newcommand{\calv}{\mathscr{V}}

\newcommand{\frp}{\mathfrak{p}}
\newcommand{\frq}{\mathfrak{q}}

\newcommand{\loc}{{\rm loc}}

\newcommand{\supp}{{\rm supp}}
\newcommand{\Id}{{\rm Id}}

\newcommand{\sfi}{{\rm sf}}

\newcommand{\lin}{{\rm lin}}

\newcommand{\mo}{{\rm mod}}
\newcommand{\ran}{{\rm ran}}

\newcommand{\la}{\langle}
\newcommand{\ra}{\rangle}

\newcommand{\bbI}{\mathbbm{1}}

\newcommand{\bdot}{\boldsymbol{\cdot}}

\newcommand{\IL}{[\![}
\newcommand{\IR}{]\!]}

\begin{document}

\title[Existence of equivalent local martingale deflators]{Existence of equivalent local martingale deflators in semimartingale market models}
\author{Eckhard Platen \and Stefan Tappe}
\address{University of Technology Sydney, School of Mathematical Sciences and Finance Discipline Group, PO Box 123, Broadway, NSW 2007, Australia}
\email{eckhard.platen@uts.edu.au}
\address{Karlsruhe Institute of Technology, Institute of Stochastics, Postfach 6980, 76049 Karlsruhe, Germany}
\email{jens-stefan.tappe@kit.edu}
\date{2 June, 2020}
\thanks{We are grateful to Claudio Fontana, Alexander Melnikov and Martin Schweizer for valuable discussions.}
\begin{abstract}
This paper offers a systematic investigation on the existence of equivalent local martingale deflators, which are multiplicative special semimartingales, in financial markets given by positive semimartingales. In particular, it shows that the existence of such deflators can be characterized by means of the modified semimartingale characteristics. Several examples illustrate our results. Furthermore, we provide interpretations of the deflators from an economic point of view.
\end{abstract}

\maketitle\thispagestyle{empty}

\section{Introduction}\label{sec-intro}

For a given model of a financial market an important issue is the absence of arbitrage opportunities. For continuous time models this has been characterized by means of the existence of an equivalent local martingale measure (ELMM); see, in particular, the papers \cite{DS-94} and \cite{DS-98}, the textbook \cite{DS-book}, and also the paper \cite{Kabanov}. The absence of arbitrage has also been characterized by means of the existence of an equivalent local martingale deflator (ELMD); see \cite{Takaoka-Schweizer}, and also the earlier papers \cite{Choulli-Stricker} and \cite{Kardaras-12}. In certain situations, results about criteria for the absence of arbitrage have been derived, for example, in \cite{Melnikov-Shiryaev, Criens-1, Criens-2}.

Moreover, if a financial market is free of arbitrage opportunities, it arises the question how to perform pricing and hedging of contingent claims. For example, if there are several ELMMs, one has to choose a suitable pricing measure. Work in this direction has been done in the risk-neutral setting, for example, in \cite{F-Schweizer, Schweizer-1992, HPS-92, Christopeit-Musiela, Schweizer-1995, DS-variance}, and beyond the risk-neutral approach, for example, in \cite{Fernholz-Karatzas, Platen, Rung-Galesso, Fontana-Rung}.

For the aforementioned results concerning the absence of arbitrage it is implicitly assumed that the market consists of discounted price processes of risky assets with respect to some savings account. In the recent paper \cite{Platen-Tappe-FTAP} we have characterized the absence of arbitrage opportunities for semimartingale models which do not need to have a savings account that could be used as num\'{e}raire. More precisely, let $\bbs = \{ S^1,\ldots,S^d \}$ be a financial market consisting of positive semimartingales which does not need to have a savings account. Provided we are allowed to add a savings account $B$ to the market, we have shown in \cite{Platen-Tappe-FTAP} that the market is free of arbitrage if and only if there exists an ELMD $Z$ which is a multiplicative special semimartingale, and that in this case the savings account has to fit into the multiplicative decomposition $Z = D B^{-1}$ of the deflator. In this context, a \emph{savings account} is a predictable, strictly positive process $B$ of locally finite variation.

Motivated by this result, the goal of this paper is to provide a systematic investigation on the existence of ELMDs which are multiplicative semimartingales, and to interpret these deflators from an economic point of view. We will study this problem for an arbitrary market $\bbs = \{ S^i : i \in I \}$ with an arbitrary nonempty index set $I \neq \emptyset$, identifying the tradeable securities. As noted in \cite{Platen-Tappe-FTAP}, the existence of an ELMD, which is a multiplicative special semimartingale, is then still sufficient for the absence of arbitrage. Note that later on in this paper we will often consider a finite market $\bbs = \{ S^1,\ldots,S^d \}$. Since we consider a market with strictly positive semimartingales, for each $i \in I$ we have
\begin{align}\label{S-intro}
S^i = S_0^i \, \cale(X^i),
\end{align}
where $S_0^i > 0$, the process $X^i$ is a semimartingale such that $X_0 = 0$ and $\Delta X^i > -1$, and $\cale$ denotes the stochastic exponential. We will assume that for each $i \in I$ the semimartingale $X^i$ (or, equivalently, the semimartingale $S^i$) is a special semimartingale with canonical decomposition
\begin{align*}
X^i = M^i + A^i,
\end{align*}
where $M^i \in \calm_{\loc}$ with $M_0^i = 0$ is the local martingale part, and the predictable process $A^i \in \calv$ is the finite variation part. Here $\calm_{\loc}$ denotes the space of local martingales, and $\calv$ denotes the space of all adapted processes with locally finite variation starting at zero. As already mentioned, we are looking for multiplicative special ELMDs of the form $Z = D B^{-1}$ with a local martingale $D$ and a, so-called, virtual savings account $B$, which would be the savings account when it were included as a traded asset in the market. Hence, we consider a multiplicative special semimartingale
\begin{align*}
Z = D B^{-1},
\end{align*}
where the local martingale $D \in \calm_{\loc}$ and the virtual savings account $B$ are given by
\begin{align*}
D = \cale(-\Theta) \quad \text{and} \quad B = \cale(R)
\end{align*}
for some $\Theta \in \calm_{\loc}$ with $\Theta_0 = 0$ and $\Delta \Theta < 1$, and some predictable $R \in \calv$ with $\Delta R > -1$. This candidate for an ELMD can be written as
\begin{align*}
Z = \cale ( -\widetilde{\Theta} - \widetilde{R} ) 
\end{align*}
with uniquely determined processes $\widetilde{\Theta} \in \calm_{\loc}$ and $\widetilde{R} \in \calv$ such that $\widetilde{\Theta}_0 = 0$ and $\Delta \widetilde{\Theta} + \Delta \widetilde{R} < 1$, where we refer to Appendix \ref{sec-mult} for further details. Our first main result (Theorem \ref{thm-special}) states that the multiplicative special semimartingale $Z$ is an ELMD if and only if for each $i \in I$ we have the drift condition
\begin{align}\label{drift-intro}
A^i = \widetilde{R} + [A^i,\widetilde{R}] + [M^i,\Theta]^p
\end{align}
satisfied, where $[M^i,\Theta]^p$ denotes the predictable compensator of the quadratic covariation $[M^i,\Theta]$. Note that (\ref{drift-intro}) provides a decomposition of the return $A^i$ of the asset $S^i$, and that the quantities appearing in this decomposition have the following interpretations:
\begin{itemize}
\item The process $\widetilde{R} + [A^i,\widetilde{R}]$ is the \emph{locally risk-free return} of the asset $S^i$ if $M^i$ were zero.

\item The process $\Theta$ is a \emph{market price of risk}.

\item Furthermore, the process $R$ is the general \emph{locally risk-free return} or \emph{virtual short rate} of the savings account $B$.
\end{itemize}
The arguments are as follows. Assume that $D$ is the density process of an equivalent probability measure $\bbq \approx \bbp$. Then the first term $\widetilde{R} + [A^i,\widetilde{R}]$ appearing in (\ref{drift-intro}) is the drift of the semimartingale $X^i$ under the measure $\bbq$; see Remark \ref{rem-decomposition} for further details. The second term $[M^i,\Theta]^p$ explains why we call $\Theta$ a market price of risk. At this stage we point out that condition (\ref{drift-intro}) is satisfied if and only if
\begin{align*}
A^i = \widetilde{R} + [A^i,\widetilde{R}] + [M^i,\widetilde{\Theta}]^p.
\end{align*}
We prefer to call $\Theta$, and not $\widetilde{\Theta}$, a market price of risk because it also shows up in the density process $D = \cale(-\Theta)$ of the measure change, provided it exists.

The drift condition (\ref{drift-intro}) also provides us with the following insight. If we consider the stocks $\bbs = \{ S^i : i \in I \}$, that is, the tradeable productive units of the economy, then a potential market price of risk $\Theta$ and also a respective virtual short rate $R$ are visible, and provided by an ELMD 
\begin{align*}
Z = \cale(-\Theta) \cale(R)^{-1}. 
\end{align*}
Of course, the processes $\Theta$ and $R$ are generally not unique. If a central bank decides to choose a lower short rate $R$, then the market price of risk $\Theta$ increases, which can potentially stimulate the economy. Accordingly, if a central bank decides to choose a higher short rate, then the market price of risk decreases, which can potentially thwart the economy. This follows from the drift condition (\ref{drift-intro}) and provides interpretations for different choices of the deflator.

The situation simplifies if for each $i \in I$ the semimartingale $X^i$ (or, equivalently, the semimartingale $S^i$) is locally square-integrable and quasi-left-continuous. Then we have $R = \widetilde{R}$ and $\Theta = \widetilde{\Theta}$, and the ELMD $Z$ can be expressed as
\begin{align*}
Z = \cale(-\Theta) \exp(R)^{-1} = \cale(-\Theta-R).
\end{align*}
Furthermore, the drift condition (\ref{drift-intro}) simplifies to
\begin{align}\label{drift-special-intro}
A^i = R + \la M^i,\Theta \ra,
\end{align}
where $\la M^i,\Theta \ra$ denotes the predictable quadratic covariation of $M^i$ and $\Theta$; see Theorem \ref{thm-special-sq} below. The quantities on the right-hand side of (\ref{drift-special-intro}) have the following interpretations:
\begin{itemize}
\item The process $R$ is the \emph{virtual short rate} and simultaneously for every $i \in I$ the \emph{locally risk-free return} of the asset $S^i$ if $M^i$ were zero.

\item The process $\Theta$ is a \emph{market price of risk}.
\end{itemize}
In order to investigate the existence of ELMDs further, consider a finite market $\bbs = \{ S^1,\ldots,S^d \}$ consisting of locally square-integrable, quasi-left-continuous semimartingales. Then the existence of an ELMD can be characterized on the basis of the modified integral characteristics of the $\bbr^d$-valued semimartingale $X = (X^1,\ldots,X^d)$ appearing in (\ref{S-intro}). More precisely, agreeing on the notation $\bbI_{\bbr^d} = (1,\ldots,1) \in \bbr^d$, the existence of an ELMD $Z$, which is a multiplicative semimartingale, is essentially equivalent to the existence of an $\bbr^d$-valued process $x$ and an $\bbr$-valued process $r$ satisfying the $\bbr^d$-valued linear equation
\begin{align}\label{eqn-mod-intro}
c_{\mo} x = a - r \bbI_{\bbr^d},
\end{align}
where the $\bbr^d$-valued process $a$ denotes the first integral characteristic, and the $\bbr^{d \times d}$-valued process $c_{\mo}$ denotes the modified second integral characteristic of $X$; we refer to Section \ref{sec-existence} for further details.

For illustration, we consider the particular situation with jump-diffusion models, where for each $i \in I$ the process $X^i$ appearing in (\ref{S-intro}) is of the form
\begin{align*}
X^i = a^i \bdot \lambda + \sigma^i \bdot W + \gamma^i * (\frp - \frq)
\end{align*}
with an $\bbr^m$-valued standard Wiener process and a homogeneous Poisson random measure $\frp$ with compensator of the form $\frq = \lambda \otimes F$, where $F$ is a $\sigma$-finite measure on the mark space $(E,\cale)$. As we will show, the existence of an ELMD $Z$, which is a multiplicative special semimartingale, is essentially equivalent to the existence of an $\bbr^m$-valued process $\theta$, an $L^2(F)$-valued process $\psi$ and an $\bbr$-valued process $r$ such that
\begin{align}\label{eqn-jd-intro}
\langle \sigma^i,\theta \rangle_{\bbr^m} + \langle \gamma^i, \psi \rangle_{L^2(F)} = a^i - r, \quad i \in I.
\end{align}
For a finite market $\bbs = \{ S^1,\ldots,S^d \}$ we can regard $\sigma$ and $\gamma$ as multidimensional linear functionals, and then equation (\ref{eqn-jd-intro}) can be expressed as the $\bbr^d$-valued linear equation
\begin{align}\label{eqn-3-var-intro}
\sigma \theta + \gamma \psi = a - r \bbI_{\bbr^d}.
\end{align}
Furthermore, we will show that the existence of a solution to (\ref{eqn-3-var-intro}) is equivalent to the existence of an $\bbr^d$-valued process $x$ and an $\bbr$-valued process $r$ satisfying the $\bbr^d$-valued linear equation (\ref{eqn-mod-intro}), where in this particular situation with a jump-diffusion model the modified second integral characteristic $c_{\mo}$ is given by
\begin{align*}
c_{\mo}^{ij} = \la \sigma^i,\sigma^j \ra_{\bbr^m} + \la \gamma^i,\gamma^j \ra_{L^2(F)} \quad \text{for all $i,j = 1,\ldots,d$.}
\end{align*}
We refer to Section \ref{sec-jd} for further details.

We provide several examples of jump-diffusion models, including Heath-Jarrow-Morton and Brody-Hughston interest rate term structure models. These two models have the interesting feature that the savings account $B$ in the multiplicative decomposition of an ELMD $Z = D B^{-1}$, provided the latter exists, is unique; see Section \ref{sec-examples} for more details.

The remainder of this paper is organized as follows. In Section \ref{sec-market} we introduce the financial market. In Section \ref{sec-examples} we present several examples, where we already utilize the results which we have indicated above. Afterwards, we proceed with the systematic investigation on the existence of ELMDs, which are special semimartingales. In Section \ref{sec-ELMD} we derive criteria when a multiplicative special semimartingale is an ELMD. In Section \ref{sec-existence} we treat the existence of ELMDs, and in Section \ref{sec-jd} we focus on jump-diffusion models. Section \ref{sec-conclusion} concludes. For convenience of the reader, several auxiliary results concerning stochastic processes, matrices and linear operators are gathered in Appendices \ref{sec-processes}--\ref{sec-matrices}.

\section{The financial market}\label{sec-market}

In this section we introduce the financial market. Let $(\Omega,\calf,(\calf_t)_{t \in \bbr_+},\bbp)$ be a stochastic basis satisfying the usual conditions, see \cite[Def. I.1.3]{Jacod-Shiryaev}. Furthermore, we assume that $\calf_0 = \{ \Omega,\emptyset \}$. Then every $\calf_0$-measurable random variable is $\bbp$-almost surely constant. Let $I \neq \emptyset$ be an arbitrary nonempty index set, and let $\bbs = \{ S^i : i \in I \}$ be a financial market consisting of positive semimartingales. More precisely, for each $i \in I$ we assume that $S^i,S_-^i > 0$. Then for each $i \in I$ we have
\begin{align}\label{stock-i}
S^i = S_0^i \, \cale(X^i)
\end{align}
with a semimartingale $X^i$ such that $X_0 = 0$ and $\Delta X^i > -1$.

\begin{definition}
We call a semimartingale $Z$ with $Z,Z_- > 0$ an \emph{equivalent local martingale deflator (ELMD)} for $\bbs$ if
\begin{align*}
S^i Z \in \calm_{\loc} \quad \text{for all $i \in I$.}
\end{align*}
\end{definition}

\begin{definition}
We call a semimartingale $\bar{Z}$ with $\bar{Z},\bar{Z}_- > 0$ an \emph{equivalent local martingale num\'{e}raire (ELMN)} for $\bbs$ if
\begin{align*}
\frac{S^i}{\bar{Z}} \in \calm_{\loc} \quad \text{for all $i \in I$.}
\end{align*}
\end{definition}

Note that a semimartingale $Z$ with $Z,Z_- > 0$ is an ELMD if and only if $\bar{Z} = Z^{-1}$ is an ELMN.

\begin{definition}
We call an equivalent probability measure $\bbq \approx \bbp$ on $(\Omega,\calf_{\infty -})$ an \emph{equivalent local martingale measure (ELMM)} for $\bbs$ if $S^i$ is a $\bbq$-local martingale for all $i \in I$.
\end{definition}

\begin{definition}
We call every predictable process $B$ of locally finite variation with $B_0 = 1$ and $B, B_- > 0$ a \emph{savings account} (or a \emph{locally risk-free asset}).
\end{definition}

The following result shows why we are interested in the existence of an ELMD which is a multiplicative special semimartingale. We denote by $\calp_{\sfi,1}^+(\bbs)$ the convex set of all nonnegative, self-financing portfolios with initial value one evaluated at a fixed terminal time $T \in (0,\infty)$. Considering this set is equivalent to looking at all nonnegative, self-financing portfolios with strictly positive initial values; see \cite{Platen-Tappe-tvs} and \cite{Platen-Tappe-FTAP} for more details. Furthermore, the following no-arbitrage concepts are NUPBR (No Unbounded Profit with Bounded Risk), NAA$_1$ (No Asymptotic Arbitrage of the 1st Kind) and NA$_1$ (No Arbitrage of the 1st Kind), which are known to be equivalent in the present situation; see, for example \cite{KKS}, \cite{Platen-Tappe-tvs} or \cite{Platen-Tappe-FTAP} for further details.

\begin{theorem}\label{thm-FTAP}
Suppose there exists an ELMD $Z$ for $\bbs$ which is a multiplicative special semimartingale, and let $Z = D B^{-1}$ be a multiplicative decomposition with a savings account $B$. Then $\calp_{\sfi,1}^+(\bbs \cup \{ B \})$ satisfies NUPBR, NAA$_1$ and NA$_1$.
\end{theorem}

\begin{proof}
This is a consequence of \cite[Thm. 7.5 and Remark 7.10]{Platen-Tappe-FTAP}.
\end{proof}

If the market $\bbs$ is finite, that is the index set $I$ is finite, then the existence of such an ELMD is equivalent to the existence of a savings account $B$ such that $\calp_{\sfi,1}^+(\bbs \cup \{ B \})$ satisfies NUPBR, NAA$_1$ and NA$_1$; see \cite[Thm. 7.5]{Platen-Tappe-FTAP}. In view of Theorem \ref{thm-FTAP}, we are interested in the existence of an ELMD $Z$ which is a multiplicative special semimartingale because this ensures the absence of arbitrage.

\section{Examples}\label{sec-examples}

Before we proceed with the systematic investigation on the existence of ELMDs, which are multiplicative special semimartingales, for the purpose of illustration we present concrete examples of financial models, where we discuss the existence of such deflators. In the upcoming examples, we utilize the results which we will develop in Sections \ref{sec-ELMD}--\ref{sec-jd} later on. In each of the following examples the market $\bbs = \{ S^i : i \in I \}$ is given by a jump-diffusion model, where for each $i \in I$ the process $X^i$ appearing in (\ref{stock-i}) is of the form
\begin{align*}
X^i = a^i \bdot \lambda + \sigma^i \bdot W + \gamma^i * (\frp - \frq)
\end{align*}
with $\lambda$ denoting the Lebesgue measure, an $\bbr^m$-valued standard Wiener process $W$ and a homogeneous Poisson random measure $\frp$ with compensator of the form $\frq = \lambda \otimes F$, where $F$ is a $\sigma$-finite measure on the mark space $(E,\cale)$. In order to look for ELMDs which are multiplicative special semimartingales, we consider a multiplicative special semimartingale $Z = D B^{-1}$, where
\begin{align}\label{D-and-B-jd-pre}
D = \cale \big( -\theta \bdot W - \psi * (\frp - \frq) \big) \quad \text{and} \quad B = \cale(r \bdot \lambda) = \exp(r \bdot \lambda)
\end{align}
with appropriate processes $\theta$, $\psi$ and $r$ such that $\psi < 1$. As already mentioned in Section \ref{sec-intro}, for the absence of arbitrage we have to find a solution to the linear equation (\ref{eqn-jd-intro}), and for a  finite market $\bbs = \{ S^1,\ldots,S^d \}$ this equation can be expressed as the $\bbr^d$-valued linear equation (\ref{eqn-3-var-intro}). We refer to Section \ref{sec-jd} for further details.

\subsection{Pure diffusion models}\label{sec-pure-diffusion}

Consider a finite market $\bbs$, where the $\bbr^d$-valued semimartingale $X$ is an It\^{o} process of the form
\begin{align*}
X = a \bdot \lambda + \sigma \bdot W.
\end{align*}
Here we consider a multiplicative special semimartingale $Z = D B^{-1}$, where the local martingale $D$ and the savings account $B$ are of the form
\begin{align}\label{D-and-B-pure-diff}
D = \cale(-\theta \bdot W) \quad \text{and} \quad B = \cale(r \bdot \lambda) = \exp(r \bdot \lambda),
\end{align}
and the $\bbr^d$-valued linear equation (\ref{eqn-3-var-intro}) reads
\begin{align}\label{eq-diffusion-1}
\sigma \theta = a - r \bbI_{\bbr^d},
\end{align}
where $\sigma$ is regarded as an $\bbr^{d \times m}$-valued process. The existence of a solution $(\theta,r)$ to (\ref{eq-diffusion-1}) gives rise to an ELMD $Z = D B^{-1}$ for the market $\bbs$ with the local martingale $D$ and the savings account $B$ given by (\ref{D-and-B-pure-diff}).

\subsection{Examples of pure diffusion models}

A particular situation of the pure diffusion model from Section \ref{sec-pure-diffusion} arises if the $\bbr^d$-valued semimartingale $X$ is of the form
\begin{align*}
X = a \bdot \lambda + \sigma \bbI_{\bbr^d} \bdot W
\end{align*}
with constants $a \in \bbr^d$, $\sigma > 0$ and an $\bbr$-valued standard Wiener process $W$. Then the $\bbr^d$-valued linear equation (\ref{eq-diffusion-1}) reads
\begin{align*}
a = (\sigma \theta + r) \bbI_{\bbr^d},
\end{align*}
and this equation has a solution if and only if $a \in \lin \{ \bbI_{\bbr^d} \}$, where $\lin \{ \bbI_{\bbr^d} \}$ denotes the one-dimensional linear space generated by $\bbI_{\bbr^d} = (1,\ldots,1) \in \bbr^d$.

\subsection{The Black-Scholes model}

In the Black-Scholes model, which goes back to \cite{Black-Scholes} and \cite{Merton-BS}, the market is given by $\bbs = \{ S \}$, where the stock price satisfies
\begin{align*}
dS_t = a S_t dt + \sigma S_t dW_t
\end{align*}
with constants $a \in \bbr$, $\sigma > 0$ and an $\bbr$-valued standard Wiener process $W$. This is a particular case of the pure diffusion models considered in Section \ref{sec-pure-diffusion}. For a fixed constant $r \in \bbr$ the linear equation (\ref{eq-diffusion-1}) is the $\bbr$-valued equation
\begin{align*}
\sigma \theta = a - r,
\end{align*}
and it has the unique solution
\begin{align*}
\theta = \frac{a - r}{\sigma}.
\end{align*}
Therefore, choosing the ELMD $Z = D B^{-1}$ with the local martingale $D$ and the savings account $B$ given by (\ref{D-and-B-pure-diff}), we deduce that $\calp_{\sfi,1}^+(\bbs \cup \{ B \}) = \calp_{\sfi,1}^+(\{ S,B \})$ satisfies NUPBR, NAA$_1$ and NA$_1$. This is in accordance with our findings from \cite[Sec. 9]{Platen-Tappe-FTAP}.

\subsection{The Heston model}

In the Heston model, which goes back to \cite{Heston}, the market is given by $\bbs = \{ S \}$, where the stock price satisfies
\begin{align*}
dS_t = a S_t dt + \sqrt{v_t} S_t dW_t
\end{align*}
with a constant $a \in \bbr$ and an $\bbr$-valued standard Wiener process $W$. The variance process $v$ is a Cox-Ingersoll-Ross process
\begin{align*}
dv_t = \kappa(\vartheta - v_t) dt + \xi \sqrt{v_t} d\widetilde{W}_t
\end{align*}
with initial value $v_0 > 0$ and constants $\kappa,\vartheta,\xi > 0$. The process $\widetilde{W}$ is another $\bbr$-valued standard Wiener process. We assume that $2 \kappa \vartheta > \xi^2$, which ensures that the variance process $v$ is strictly positive. The Heston model is also a particular case of the pure diffusion models considered in Section \ref{sec-pure-diffusion}. For a fixed constant $r \in \bbr$, the linear equation (\ref{eq-diffusion-1}) is the $\bbr$-valued equation
\begin{align*}
\sqrt{v} \theta = a - r,
\end{align*}
and it has the unique solution
\begin{align*}
\theta = \frac{a - r}{\sqrt{v}}.
\end{align*}
By the continuity of $v$ we have $\theta \in L_{\loc}^2(W)$. Therefore, choosing the ELMD $Z = D B^{-1}$ with the local martingale $D$ and the savings account $B$ given by (\ref{D-and-B-pure-diff}), we deduce that $\calp_{\sfi,1}^+(\bbs \cup \{ B \}) = \calp_{\sfi,1}^+(\{ S,B \})$ satisfies NUPBR, NAA$_1$ and NA$_1$.

\subsection{The Merton model}\label{subsec-Merton}

In the Merton model, which goes back to \cite{Merton}, the market is given by $\bbs = \{ S \}$, where the stock price satisfies
\begin{align*}
dS_t = \mu S_t dt + \sigma S_t dW_t + S_{t-} dQ_t
\end{align*}
with constants $\mu \in \bbr$, $\sigma > 0$, an $\bbr$-valued standard Wiener process $W$ and a compound Poisson process $Q$ such that $\Delta Q > -1$. More precisely, the jump size distribution is that of $Y-1$, where $Y$ has a lognormal distribution. We consider a more general situation, namely a jump-diffusion model with one asset, where the mark space is given by
\begin{align*}
(E,\cale) = (\bbr,\calb(\bbr)), 
\end{align*}
and where the measure $F$ satisfies 
\begin{align*}
F(\bbr) < \infty, \quad \supp(F) \subset (-1,\infty) \quad \text{and} \quad \int_{(-1,\infty)} |x|^2 F(dx) < \infty. 
\end{align*}
Furthermore, we define the mapping $\gamma : \Omega \times \bbr_+ \times \bbr \to \bbr$ as
\begin{align*}
\gamma(\omega,t,x) = x \bbI_{\{ x > -1 \}}, \quad (\omega,t,x) \in \Omega \times \bbr_+ \times \bbr.
\end{align*}
Note that this model covers the Merton model, because the lognormal distribution admits second order moments. The linear equation (\ref{eqn-3-var-intro}) is the $\bbr$-valued equation
\begin{align}\label{eq-Merton}
\sigma \theta + \int_{\bbr} x \psi(x) F(dx) = a - r.
\end{align}
Note that equation (\ref{eq-Merton}) admits several solutions. For example, choose a constant $r \in \bbr$ and a function $\psi \in L^2(F)$ such that $\psi < 1$. Then the solution to the linear equation (\ref{eq-Merton}) is given by
\begin{align*}
\theta = \frac{1}{\sigma} \bigg( a - r - \int_{\bbr} x \psi(x) F(dx) \bigg).
\end{align*}
Therefore, choosing the ELMD $Z = D B^{-1}$ with the local martingale $D$ and the savings account $B$ given by (\ref{D-and-B-jd-pre}), we deduce that $\calp_{\sfi,1}^+(\bbs \cup \{ B \}) = \calp_{\sfi,1}^+(\{ S,B \})$ satisfies NUPBR, NAA$_1$ and NA$_1$.

\subsection{Jump-diffusion models with finitely many jumps}\label{subsec-jd-fin}

Consider a finite market $\bbs$ and assume that the measure $F$ on the mark space $(E,\cale)$ is concentrated on finitely many points. More precisely, we assume there are pairwise different elements $x_1,\ldots,x_n \in E$ such that $\{ x_j \} \in \cale$ for each $j=1,\ldots,n$, and that the measure $F$ is of the form
\begin{align*}
F = \sum_{j=1}^n c_j \delta_{x_j}
\end{align*}
with finite numbers $c_1,\ldots,c_n > 0$. Here $\delta_{x_j}$ denotes the Dirac measure at point $x_j$ for each $j=1,\ldots,n$. Then the space $L_{\loc}^2(\frp)$ can be identified with the space of all optional processes $\rho : \Omega \times \bbr_+ \to \bbr^n$ such that $\| \rho \|_{\bbr^n}^2 \bdot \lambda \in \calv^+$. Indeed, for each $\psi \in L_{\loc}^2(\frp)$ the corresponding process $\rho$ is given by
\begin{align*}
\rho = \big( \psi(x_1),\ldots,\psi(x_n) \big).
\end{align*}
With this identification we have $\psi < 1$ if and only if $\rho^j < 1$ for all $j=1,\ldots,n$, and for each $\psi \in L_{\loc}^2(\frp)$ we have
\begin{align*}
\gamma \psi = \Gamma \rho,
\end{align*}
where the $\bbr^{d \times n}$-valued process $\Gamma$ is given by
\begin{align*}
\Gamma^{ij} = c_j \gamma^i(x_j) \quad \text{for all $i=1,\ldots,d$ and $j=1,\ldots,n$.}
\end{align*}
Therefore, for $\theta \in L_{\loc}^2(W)$ and an optional processes $\rho : \Omega \times \bbr_+ \to \bbr^n$ such that $\| \rho \|_{\bbr^n}^2 \bdot \lambda \in \calv^+$ and $\rho^j < 1$ for all $j=1,\ldots,n$ the multiplicative special semimartingale $Z = D B^{-1}$ with the local martingale $D$ and the savings account $B$ given by (\ref{D-and-B-jd-pre}), where $\psi \in L_{\loc}^2(\frp)$ is defined as $\psi := \sum_{j=1}^n \rho^j \bbI_{\{ x_j \}}$, is an ELMD for $\bbs$ if and only if
\begin{align}\label{eq-jd-fin}
\sigma \theta + \Gamma \rho = a - r \bbI_{\bbr^d}.
\end{align}
Recall that in the linear equation (\ref{eq-jd-fin}) the process $\sigma$ is $\bbr^{d \times m}$-valued, and that the process $\Gamma$ is $\bbr^{d \times n}$-valued.

\subsection{The Black-Scholes model with an additional Poisson process}

Consider a market $\bbs = \{ S \}$ with one asset, where the stock price satisfies
\begin{align*}
dS_t = \mu S_t dt + \sigma S_t dW_t + S_{t-} dN_t
\end{align*}
with constants $\mu \in \bbr$, $\sigma > 0$, an $\bbr$-valued standard Wiener process $W$ and a Poisson process $N$ with intensity $c > 0$. Note that this is a particular case of the Merton type model considered in Section \ref{subsec-Merton}. This model is also a jump-diffusion model as in Section \ref{subsec-jd-fin}, and the linear equation (\ref{eq-jd-fin}) is the $\bbr$-valued equation
\begin{align}\label{eq-BS-Poisson}
\sigma \theta + c \rho = a - r.
\end{align}
Clearly, there are several solutions $(\theta,\rho,r)$ with $\rho < 1$ to the linear equation (\ref{eq-BS-Poisson}), ensuring the absence of arbitrage. For example, for constants $r \in \bbr$ and $\rho < 1$ the solution to (\ref{eq-BS-Poisson}) is given by
\begin{align*}
\theta = \frac{a-r-c\rho}{\sigma}.
\end{align*}

\subsection{Heath-Jarrow-Morton interest rate term structure models}

In this section we consider Heath-Jarrow-Morton (HJM) interest term structure models for modeling markets of zero coupon bonds. Under risk-neutral pricing we refer to \cite{HJM} for the classical HJM model driven by Wiener processes. Furthermore, we refer, for example, to \cite{Eberlein-Raible, Eberlein_J} for risk-neutral HJM models driven by L\'{e}vy processes, and, for example, to \cite{BKR0, BKR, FTT-published} for risk-neutral HJM models driven by Wiener processes and Poisson random measures. In the framework of the Benchmark Approach (see \cite{Platen}) HJM models driven by Wiener processes and Poisson random measures have been studied in \cite{Christensen-Platen, B-N-Platen}. We assume that for each $T \in \bbr_+$ the forward rate is given by
\begin{align*}
f(T) = f_0(T) + \alpha(T) \bdot \lambda + \sigma(T) \bdot W + \gamma(T) * (\frp - \frq)
\end{align*}
with a starting value $f_0(T) \in \bbr$ and suitable integrands $\alpha(T) \in L_{\loc}^1(\lambda)$, $\sigma(T) \in L_{\loc}^2(W)$ and $\gamma(T) \in L_{\loc}^2(\frp)$. By a monotone class argument, the short rate $f_{\cdot}(\cdot)$ given by
\begin{align*}
f_t(t) \quad \text{for all $t \in \bbr_+$}
\end{align*}
has an optional version. for each $T \in \bbr_+$ we define the process $F(T)$ as
\begin{align*}
F_t(T) := -\int_t^T f_t(s) ds, \quad t \in \bbr_+,
\end{align*}
and the bond price process
\begin{align*}
P(T) := \exp(F(T)).
\end{align*}
In the sequel, we are interested in the bond market
\begin{align*}
\bbs = \{ P(T) : T \in \bbr_+ \}.
\end{align*}
Let $T \in \bbr_+$ be arbitrary. We define the processes $A(T)$, $\Sigma(T)$ and $\Gamma(T)$ as follows. For $t \leq T$ we set
\begin{align*}
A_t(T) := -\int_t^T \alpha_t(s) ds, \quad \Sigma_t(T) := -\int_t^T \sigma_t(s) ds, \quad \Gamma_t(T) := -\int_t^T \gamma_t(s) ds,
\end{align*}
and for $t > T$ we set
\begin{align*}
A_t(T) := A_T(T), \quad \Sigma_t(T) := \Sigma_T(T), \quad \Gamma_t(T) := \Gamma_T(T).
\end{align*}
Of course, this requires appropriate integrability conditions on $\alpha$, $\sigma$ and $\gamma$.
Under further suitable regularity conditions, a standard calculation using stochastic Fubini theorems shows that
\begin{align*}
F(T) = F_0(T) + (f_{\cdot}(\cdot) + A(T)) \bdot \lambda + \Sigma(T) \bdot W + \Gamma(T) * (\frp - \frq).
\end{align*}
By \cite[Thm. II.8.10]{Jacod-Shiryaev} we have
\begin{align*}
P(T) = P_0(T) \, \cale(X(T)),
\end{align*}
where the semimartingale $X(T)$ is given by
\begin{align*}
X(T) &= \bigg( f_{\cdot}(\cdot) + A(T) + \frac{1}{2} \| \Sigma(T) \|_{\bbr^m}^2 + \int_E \big( e^{\Gamma(T)} - 1 - \Gamma(T) \big) F(dx) \bigg) \bdot \lambda
\\ &\quad + \Sigma(T) \bdot W + \big( e^{\Gamma(T)} - 1 \big) * (\frp - \frq).
\end{align*}
Now, we consider a multiplicative special semimartingale $Z = D B^{-1}$, where the local martingale $D$ and the savings account $B$ are given by (\ref{D-and-B-jd-pre}) with $\theta \in L_{\loc}^2(W)$ and $\psi \in L_{\loc}^2(\frp)$ such that $\psi < 1$, as well as an optional process $r \in L_{\loc}^1(\lambda)$.

\begin{theorem}
We assume that the processes $f_{\cdot}(\cdot)$, $r$, $\theta$ and $\psi$ are c\`{a}d (right-continuous) or c\`{a}g (left-continuous), and that for each $T \in \bbr_+$ the processes $A(T)$, $\Sigma(T)$ and $\Gamma(T)$ are c\`{a}d or c\`{a}g. Then the following statements are equivalent:
\begin{enumerate}
\item[(i)] $Z$ is an ELMD for the bond market $\bbs$.

\item[(ii)] $Z$ is an ELMD for the extended bond market $\bbs \cup \{ B \}$.

\item[(iii)] We have up to an evanescent set
\begin{align}\label{short-rate-f}
r = f_{\cdot}(\cdot)
\end{align}
and for each $T \in \bbr_+$ we have up to an evanescent set
\begin{equation}\label{HJM-drift}
\begin{aligned}
-A(T) &= \frac{1}{2} \| \Sigma(T) \|_{\bbr^m}^2 - \la \Sigma(T),\theta \ra_{\bbr^m}
\\ &\quad + \int_E \Big( (1 - \psi(x)) ( e^{\Gamma(T,x)} - 1 ) - \Gamma(T,x) \Big) F(dx).
\end{aligned}
\end{equation}
\end{enumerate}
If the previous conditions are fulfilled, then $\calp_{\sfi,1}^+(\bbs \cup \{ B \})$ satisfies NUPBR, NAA$_1$ and NA$_1$.
\end{theorem}

\begin{proof}
Noting the assumed regularity conditions, by Theorem \ref{thm-jd} the process $Z$ is an ELMD for $\bbs$, or equivalently for $\bbs \cup \{ B \}$, if and only if for each $T \in \bbr_+$ we have up to an evanescent set
\begin{align*}
&\la \Sigma(T),\theta \ra_{\bbr^m} + \la \Gamma(T),\psi \ra_{L^2(F)}
\\ &= f_{\cdot}(\cdot) + A(T) + \frac{1}{2} \| \Sigma(T) \|_{\bbr^m}^2 + \int_E \big( e^{\Gamma(T,x)} - 1 - \Gamma(T,x) \big) F(dx) - r.
\end{align*}
Evaluating this equation at $t=T$ for every $T \in \bbr_+$ we obtain that $Z$ is an ELMD for $\bbs$ if and only if we have (\ref{short-rate-f}) up to an evanescent set, and for each $T \in \bbr_+$ we have (\ref{HJM-drift}) up to an evanescent set. The additional statement is a consequence of Theorem \ref{thm-FTAP}.
\end{proof}

Consequently, if an ELMD $Z$ for the bond market, which is a multiplicative special semimartingale, exists, then the savings account $B$ in the multiplicative decomposition $Z = D B^{-1}$ is unique, and it is given by
\begin{align*}
B_t = \exp \bigg( \int_0^t f_s(s) ds \bigg), \quad t \in \bbr_+.
\end{align*}

\begin{remark}
Differentiating equation (\ref{HJM-drift}) with respect to $T$ yields the drift condition
\begin{align*}
\alpha(T) &= - \la \sigma(T), \Sigma(T) - \theta \ra_{\bbr^m} - \int_E \gamma(T,x) \big( (1 - \psi(x)) e^{\Gamma(T,x)} - 1 \big) F(dx).
\end{align*}
This drift condition also appears in the framework of the Benchmark Approach; see \cite{Christensen-Platen} and \cite{B-N-Platen}.
\end{remark}

\subsection{Brody-Hughston interest rate term structure models}

In this section we investigate Brody-Hughston interest rate term structure models. Such term structure models driven by Wiener processes have been introduced in \cite{Brody-Hughston-0, Brody-Hughston} in the risk-neutral setting; see also \cite{FTT-Positivity} for such models driven by Wiener processes and Poisson random measures in the risk-neutral setting. Let $\rho = (\rho_t)_{t \in \bbr_+}$ be a stochastic process consisting of strictly positive probability densities on $\bbr_+$. For each $T \in \bbr_+$ we define the bond prices
\begin{align}\label{bond-dens}
P_t(T) := \int_{T-t}^{\infty} \rho_t(u) du, \quad t \in [0,T].
\end{align}
We are interested in the bond market
\begin{align*}
\bbs = \{ P(T) : T \in \bbr_+ \}.
\end{align*}
Let $T \in \bbr_+$ be arbitrary. Since the bond prices given by (\ref{bond-dens}) are strictly positive, we may assume that
\begin{align*}
P(T) = P_0(T) \, \cale \big( \alpha(T) \bdot \lambda + \sigma(T) \bdot W + \gamma(T) * (\frp - \frq) \big)
\end{align*}
with suitable integrands $\alpha(T) \in L_{\loc}^1(\lambda)$, $\sigma(T) \in L_{\loc}^2(W)$ and $\gamma(T) \in L_{\loc}^2(\frp)$. Then we have
\begin{align*}
P(T) = P_0(T) + P(T) \alpha(T) \bdot \lambda + P(T) \sigma(T) \bdot W + P_-(T) \gamma(T) * (\frp - \frq).
\end{align*}
Now we switch to the Musiela parametrization
\begin{align*}
p_t(x) = P_t(t+x), \quad t,x \in \bbr_+. 
\end{align*}
Subject to appropriate regularity conditions, we obtain that $p$ is a solution to the stochastic partial differential equation (SPDE)
\begin{align*}
p(x) &= p_0(x) + \big( \partial_x p(x) + p(x) \hat{\alpha}(x) \big) \bdot \lambda + p(x) \hat{\sigma}(x) \bdot W
\\ &\quad + p_-(x) \hat{\gamma}(x) * ( \frp - \frq ), \quad x \in \bbr_+,
\end{align*}
where the new coefficients $\hat{\alpha}, \hat{\sigma}, \hat{\gamma}$ are given by
\begin{align*}
\hat{\alpha}_t(x) := \alpha(t,t+x), \quad t,x \in \bbr_+,
\\ \hat{\sigma}_t(x) := \sigma(t,t+x), \quad t,x \in \bbr_+,
\\ \hat{\gamma}_t(x) := \gamma(t,t+x), \quad t,x \in \bbr_+.
\end{align*}
By (\ref{bond-dens}) we have
\begin{align}\label{bond-Musiela}
p_t(x) = \int_x^{\infty} \rho_t(u)du, \quad t,x \in \bbr_+,
\end{align}
and hence
\begin{align*}
\rho_t(x) = -\partial_x p_t(x), \quad t,x \in \bbr_+.
\end{align*}
Therefore, subject to appropriate regularity conditions, which allow us to interchange differentiation and integration, we obtain that $\rho$ satisfies the SPDE
\begin{align*}
\rho(x) &= \rho_0(x) + \big( \partial_x \rho(x) - \partial_x ( p(x) \hat{\alpha}(x) ) \big) \bdot \lambda
\\ &\quad - \partial_x ( p(x) \hat{\sigma}(x) ) \bdot W - \partial_x ( p_-(x) \hat{\gamma}(x) ) * ( \frp - \frq ), \quad x \in \bbr_+.
\end{align*}
Since $\rho$ is strictly positive, we can define the coefficients $\bar{\alpha}, \bar{\sigma}, \bar{\gamma}$ as
\begin{align*}
\bar{\alpha}(x) &:= -\frac{\partial_x (p(x) \hat{\alpha}(x))}{\rho(x)}, \quad x \in \bbr_+,
\\ \bar{\sigma}(x) &:= -\frac{\partial_x (p(x) \hat{\sigma}(x))}{\rho(x)}, \quad x \in \bbr_+,
\\ \bar{\gamma}(x) &:= -\frac{\partial_x (p_-(x) \hat{\gamma}(x))}{\rho(x)}, \quad x \in \bbr_+.
\end{align*}
Therefore, we obtain
\begin{align*}
\rho(x) &= \rho_0(x) + \big( \partial_x \rho(x) + \rho(x) \bar{\alpha}(x) \big) \bdot \lambda
\\ &\quad + \rho(x) \bar{\sigma}(x) \bdot W + \rho_-(x) \bar{\gamma}(x) * ( \frp - \frq ), \quad x \in \bbr_+.
\end{align*}
Consequently, noting that
\begin{align*}
\rho(x) \partial_x \ln ( \rho(x) ) = \partial_x \rho(x), \quad x \in \bbr_+,
\end{align*}
for each $x \in \bbr_+$ we have the representation
\begin{align*}
\rho(x) = \rho_0(x) \, \cale \big( ( \partial_x \ln ( \rho(x) ) + \bar{\alpha}(x) ) \bdot \lambda +  \bar{\sigma}(x) \bdot W + \bar{\gamma}(x) * ( \frp - \frq ) \big).
\end{align*}
Since the process $\rho$ leaves the convex set of probability densities invariant, we have up to an evanescent set
\begin{align}\label{cons-1}
\int_0^{\infty} \big( \partial_x \rho(x) + \rho(x) \bar{\alpha}(x) \big) dx &= 0,
\\ \label{cons-2} \int_0^{\infty} \bar{\sigma}(x) \rho(x) dx &= 0,
\\ \label{cons-3} \int_0^{\infty} \bar{\gamma}(x) \rho_-(x) dx &= 0.
\end{align}
Note that condition (\ref{cons-1}) is satisfied if and only if we have up to an evanescent set
\begin{align}\label{rho-zero}
\rho(0) = \int_0^{\infty} \rho(x) \bar{\alpha}(x) dx \quad \text{$\lambda$-a.e.} \quad \text{$\bbp$-a.e.}
\end{align}
Now, we consider a multiplicative special semimartingale $Z = D B^{-1}$, where the local martingale $D$ and the savings account $B$ are given by (\ref{D-and-B-jd-pre}) with $\theta \in L_{\loc}^2(W)$ and $\psi \in L_{\loc}^2(\frp)$ such that $\psi < 1$, as well an an optional process $r \in L_{\loc}^1(\lambda)$.

\begin{theorem}
We assume that the following conditions are fulfilled:
\begin{itemize}
\item The processes $r$, $\theta$ and $\psi$ are c\`{a}d or c\`{a}g, and for each $x \in \bbr_+$ the process $\rho(x)$ is c\`{a}d or c\`{a}g.

\item For each $T \in \bbr_+$ the processes $\alpha(T)$, $\sigma(T)$ and $\gamma(T)$ are c\`{a}d or c\`{a}g, and the processes $\alpha$, $\sigma$ and $\gamma$ are continuous in the second argument $T$.

\item For each $x \in \bbr_+$ the processes $\hat{\alpha}(x)$, $\hat{\sigma}(x)$ and $\hat{\gamma}(x)$ are c\`{a}d or c\`{a}g, and the processes $\hat{\alpha}$, $\hat{\sigma}$ and $\hat{\gamma}$ are continuous in the second argument $x$.
 
\item For each $x \in \bbr_+$ the processes $\bar{\alpha}(x)$, $\bar{\sigma}(x)$ and $\bar{\gamma}(x)$ are c\`{a}d or c\`{a}g.
\end{itemize}
Then the following statements are equivalent:
\begin{enumerate}
\item[(i)] $Z$ is an ELMD for the bond market $\bbs$.

\item[(ii)] $Z$ is an ELMD for the extended bond market $\bbs \cup \{ B \}$.

\item[(iii)] We have up to an evanescent set
\begin{align}\label{short-rate-BH}
r = \rho(0),
\end{align}
and for each $x \in \bbr_+$ we have up to an evanescent set
\begin{align}\label{drift-BH-model}
\bar{\alpha}(x) = r + \la \bar{\sigma}(x),\theta \ra_{\bbr^m} + \la \bar{\gamma}(x),\psi \ra_{L^2(F)}.
\end{align}
\end{enumerate}
If the previous conditions are fulfilled, then $\calp_{\sfi,1}^+(\bbs \cup \{ B \})$ satisfies NUPBR, NAA$_1$ and NA$_1$.
\end{theorem}

\begin{proof}
In view of the assumed regularity, by Theorem \ref{thm-jd} the process $Z$ is an ELMD for $\bbs$, or equivalently for $\bbs \cup \{ B \}$, if and only if for each $T \in \bbr_+$ we have up to an evanescent set
\begin{align}\label{drift-BH-bonds}
\alpha(T) = r + \la \sigma(T),\theta \ra_{\bbr^m} + \la \gamma(T),\psi \ra_{L^2(F)}.
\end{align}
By the assumed continuity in the second argument, this is satisfied if and only if for each $x \in \bbr_+$ we have up to an evanescent set
\begin{align}\label{drift-SPDE}
\hat{\alpha}(x) = r + \la \hat{\sigma}(x),\theta \ra_{\bbr^m} + \la \hat{\gamma}(x),\psi \ra_{L^2(F)}.
\end{align}
Let $x \in \bbr_+$ be arbitrary. If condition (\ref{drift-SPDE}) is satisfied, then we have
\begin{align*}
\bar{\alpha}(x) &= -\frac{\partial_x (p(x) \hat{\alpha}(x))}{\rho(x)} = -\frac{\partial_x (p(x) (r + \la \hat{\sigma}(x),\theta \ra_{\bbr^m} + \la \hat{\gamma}(x),\psi \ra_{L^2(F)}))}{\rho(x)}
\\ &= r + \la \bar{\sigma}(x),\theta \ra_{\bbr^m} + \la \bar{\gamma}(x),\psi \ra_{L^2(F)},
\end{align*}
showing (\ref{drift-BH-model}). Conversely, suppose that condition (\ref{drift-BH-model}) is satisfied. Then we have
\begin{align*}
\partial_x (p(x) \hat{\alpha}(x)) = \partial_x \big( p(x) (r + \la \hat{\sigma}(x),\theta \ra_{\bbr^m} + \la \hat{\gamma}(x),\psi \ra_{L^2(F)}) \big).
\end{align*}
Noting (\ref{cons-1})--(\ref{cons-3}) and (\ref{bond-Musiela}), integrating gives us
\begin{align*}
p(x) \hat{\alpha}(x) &= -\int_x^{\infty} \partial_u ( p(u) \hat{\alpha}(u) ) du
\\ &= -\int_x^{\infty} \partial_u \big( p(u) (r + \la \hat{\sigma}(u),\theta \ra_{\bbr^m} + \la \hat{\gamma}(u),\psi \ra_{L^2(F)}) \big) du
\\ &= p(x) (r + \la \hat{\sigma}(x),\theta \ra_{\bbr^m} + \la \hat{\gamma}(x),\psi \ra_{L^2(F)}),
\end{align*}
showing (\ref{drift-SPDE}). Therefore, for each $x \in \bbr_+$ condition (\ref{drift-BH-model}) is satisfied up to an evanescent set if and only if for each $T \in \bbr_+$ condition (\ref{drift-BH-bonds}) is satisfied up to an evanescent set. Consequently, the process $Z$ is an ELMD for $\bbs$, or equivalently for $\bbs \cup \{ B \}$, if and only if for each $x \in \bbr_+$ we have (\ref{drift-BH-model}) up to an evanescent set. Inserting (\ref{drift-BH-model}) into (\ref{rho-zero}) and noting (\ref{cons-2}), (\ref{cons-3}) as well as $\int_0^{\infty} \rho(x) dx = 1$, we obtain that $Z$ is an ELMD for $\bbs$ if and only if we have (\ref{short-rate-BH}) and (\ref{drift-BH-model}). The additional statement is a consequence of Theorem \ref{thm-FTAP}.
\end{proof}

Note that the situation with the Brody-Hughston model is similar to that with the HJM model from the previous section. Namely, if an ELMD $Z$, which is a multiplicative special semimartingale, exists, then the savings account $B$ in the multiplicative decomposition $Z = D B^{-1}$ is unique, and it is given by
\begin{align*}
B_t = \exp \bigg( \int_0^t \rho_s(0) ds \bigg), \quad t \in \bbr_+.
\end{align*}

\section{Equivalent local martingale deflators}\label{sec-ELMD}

After these examples, we proceed with the systematic investigation on the existence of ELMDs, which are special semimartingales. In this section we derive criteria when a multiplicative special semimartingale is an ELMD, and draw some consequences. Let $\bbs = \{ S^i : i \in I \}$ be a financial market of the form (\ref{stock-i}) as in Section \ref{sec-market}. As motivated there, we are interested in the existence of an ELMD $Z$ which is a multiplicative special semimartingale because this ensures the absence of arbitrage. Let $Z$ be a semimartingale of the form
\begin{align}\label{deflator-candidate}
Z = \cale(-Y)
\end{align}
with a semimartingale $Y$ such that $Y_0 = 0$ and $\Delta Y < 1$.

\begin{proposition}\label{prop-E-formel}
The following statements are equivalent:
\begin{enumerate}
\item[(i)] $Z$ is an ELMD for $\bbs$.

\item[(ii)] For each $i \in I$ we have
\begin{align}\label{X-Y-M-loc}
X^i - Y - [X^i,Y] \in \calm_{\loc}.
\end{align}
\end{enumerate}
\end{proposition}

\begin{proof}
Let $i \in I$ be arbitrary. By Yor's formula (see \cite[II.8.19]{Jacod-Shiryaev}) we have
\begin{align*}
S^i Z = S_0^i \, \cale(X^i) \cale(-Y) = S_0^i \, \cale(X^i - Y - [X^i,Y]),
\end{align*}
which proves the stated equivalence.
\end{proof}

In the upcoming results $\cals_p$ denotes the space of all special semimartingales, and $\cala_{\loc}$ denotes the space of all elements from $\calv$ which are locally integrable; cf. \cite{Jacod-Shiryaev}.

\begin{corollary}
Suppose that $Z$ is an ELMD for $\bbs$, and that $S^i \in \cals_p$ for some $i \in I$. Then the following statements are equivalent:
\begin{enumerate}
\item[(i)] $Z$ is a multiplicative special semimartingale.

\item[(ii)] We have $Z \in \cals_p$.

\item[(iii)] We have $[X^i,Y] \in \cala_{\loc}$.
\end{enumerate}
\end{corollary}

\begin{proof}
(i) $\Leftrightarrow$ (ii): This equivalence is a consequence of \cite[Thm. II.8.21]{Jacod-Shiryaev}.

\noindent (ii) $\Leftrightarrow$ (iii): By Lemma \ref{lemma-stoch-exp-properties} we have $X^i \in \cals_p$, and we have $Z \in \cals_p$ if and only if $Y \in \cals_p$. By Proposition \ref{prop-E-formel} we have (\ref{X-Y-M-loc}). Since $X^i \in \cals_p$, we deduce that $Y + [X^i,Y] \in \cals_p$. Therefore, we have $Y \in \cals_p$ if and only if $[X^i,Y] \in \cala_{\loc}$.
\end{proof}

\begin{corollary}
Suppose that $Z$ is an ELMD for $\bbs$, which is a multiplicative special semimartingale. Then for each $i \in I$ the following statements are equivalent:
\begin{enumerate}
\item[(i)] We have $S^i \in \cals_p$.

\item[(ii)] We have $[X^i,Y] \in \cala_{\loc}$.
\end{enumerate}
\end{corollary}

\begin{proof}
According to \cite[Thm. II.8.21]{Jacod-Shiryaev} we have $Z \in \cals_p$. Hence, by Lemma \ref{lemma-stoch-exp-properties} we have $Y \in \cals_p$, and we have $S^i \in \cals_p$ if and only if $X^i \in \cals_p$. By Proposition \ref{prop-E-formel} we have (\ref{X-Y-M-loc}). Since $Y \in \cals_p$, we deduce that $X^i - [X^i,Y] \in \cals_p$. Therefore, we have $X^i \in \cals_p$ if and only if $[X^i,Y] \in \cala_{\loc}$.
\end{proof}

From now on, we assume that for each $i \in I$ the semimartingale $X^i$ appearing in (\ref{stock-i}) is a special semimartingale with canonical decomposition
\begin{align}\label{deco-1}
X^i = M^i + A^i,
\end{align}
where $M^i$ is the local martingale part and $A^i$ is the finite variation part. Furthermore, let $R \in \calv$ be a predictable process with $\Delta R > -1$, and let $\Theta \in \calm_{\loc}$ be a local martingale with $\Theta_0 = 0$ and $\Delta \Theta < 1$. Let $\widetilde{R} \in \calv$ be the predictable process with $\Delta \widetilde{R} < 1$ according to Proposition \ref{prop-bij-R}, and let $\widetilde{\Theta} \in \calm_{\loc}$ be the local martingale with $\widetilde{\Theta}_0 = 0$ and $\Delta \widetilde{\Theta} + \Delta \widetilde{R} < 1$ according to Proposition \ref{prop-bij-Theta}. We assume that the semimartingale $Y$ appearing in (\ref{deflator-candidate}) is the special semimartingale with canonical decomposition
\begin{align}\label{deco-2}
Y = \widetilde{\Theta} + \widetilde{R}.
\end{align}
Then by Proposition \ref{prop-bij-Theta} we have the multiplicative decomposition
\begin{align*}
Z = D B^{-1},
\end{align*}
where $D = \cale(-\Theta)$ and $B = \cale(R)$. We call $R$ the \emph{locally risk-free return} of the savings account $B$.

\begin{lemma}\label{lemma-special}
For each $i \in I$ the following statements are equivalent:
\begin{enumerate}
\item[(i)] We have $[X^i,Y] \in \cala_{\loc}$.

\item[(ii)] We have $[M^i,\widetilde{\Theta}] \in \cala_{\loc}$.

\item[(iii)] We have $[M^i,\Theta] \in \cala_{\loc}$.
\end{enumerate}
In either case we have
\begin{align*}
[X^i,Y]^p = [A^i,\widetilde{R}] + [M^i,\widetilde{\Theta}]^p = [A^i,\widetilde{R}] + [M^i,\Theta]^p.
\end{align*}
\end{lemma}

\begin{proof}
This is an immediate consequence of Lemma \ref{lemma-cov-equivalence}.
\end{proof}

\begin{theorem}\label{thm-special}
The following statements are equivalent:
\begin{enumerate}
\item[(i)] $Z$ is an ELMD for $\bbs$.

\item[(ii)] $Z$ is an ELMD for $\bbs \cup \{ B \}$.

\item[(iii)] For each $i \in I$ we have $[X^i,Y] \in \cala_{\loc}$ and up to an evanescent set
\begin{align}\label{drift-R-Theta}
A^i = \widetilde{R} + [ X^i,Y ]^p.
\end{align}

\item[(iv)] For each $i \in I$ we have $[M^i,\Theta] \in \cala_{\loc}$ and up to an evanescent set
\begin{align}\label{drift-R-Theta-2}
A^i = \widetilde{R} + [A^i,\widetilde{R}] + [M^i,\Theta]^p.
\end{align}
\end{enumerate}
If the previous conditions are fulfilled, then $\calp_{\sfi,1}^+(\bbs \cup \{ B \})$ satisfies NUPBR, NAA$_1$ and NA$_1$.
\end{theorem}

\begin{proof}
(i) $\Rightarrow$ (ii): This implication follows, because $BZ = D \in \calm_{\loc}$.

\noindent(ii) $\Rightarrow$ (i): This implication is obvious.

\noindent(i) $\Leftrightarrow$ (iii): Let $i \in I$ be arbitrary. Taking into account the canonical decompositions (\ref{deco-1}) and (\ref{deco-2}), we have condition (\ref{X-Y-M-loc}) if and only if
\begin{align*}
M^i + A^i - \widetilde{\Theta} - \widetilde{R} - [X^i,Y] \in \calm_{\loc},
\end{align*}
which is equivalent to
\begin{align*}
A^i - \widetilde{R} - [X^i,Y] \in \calm_{\loc} \cap \calv.
\end{align*}
Since $\calm_{\loc} \cap \calv \subset \cala_{\loc}$, this is satisfied if and only if $[X^i,Y] \in \cala_{\loc}$ and
\begin{align}\label{drift-cond-proof}
A^i - \widetilde{R} - [X^i,Y]^p \in \calm_{\loc} \cap \calv.
\end{align}
Note that the process on the left-hand side of (\ref{drift-cond-proof}) is predictable. Hence, according to \cite[Cor. I.3.16]{Jacod-Shiryaev}, condition (\ref{drift-cond-proof}) is satisfied if and only if we have (\ref{drift-R-Theta}) up to an evanescent set. Consequently, applying Proposition \ref{prop-E-formel} the process $Z$ is an ELMD for $\bbs$ if and only if we have $[X^i,Y] \in \cala_{\loc}$ and (\ref{drift-R-Theta}) up to an evanescent set.

\noindent(iii) $\Leftrightarrow$ (iv): Let $i \in I$ be arbitrary. By Lemma \ref{lemma-special} we have  $[X^i,Y] \in \cala_{\loc}$ if and only if $[M^i,\Theta] \in \cala_{\loc}$, and in this case, using Lemmas \ref{lemma-cov-R-Theta} and \ref{lemma-cov-3-fold} we obtain
\begin{align*}
[X^i,Y]^p &= [M^i+A^i,\widetilde{\Theta}+\widetilde{R}]^p = [M^i,\widetilde{\Theta}]^p + [A^i,\widetilde{R}]^p
\\ &= [M^i,\Theta - [\Theta,\widetilde{R}]]^p + [A^i,\widetilde{R}] = [M^i,\Theta]^p + [A^i,\widetilde{R}].
\end{align*}
The additional statement is a consequence of Theorem \ref{thm-FTAP}.
\end{proof}

\begin{remark}
In the situation of Theorem \ref{thm-special} we can also formally check that the drift conditions (\ref{drift-R-Theta-2}) are satisfied for the extended market $\bbs \cup \{ B \}$. Indeed, then the additional drift condition
\begin{align*}
R = \widetilde{R} + [R,\widetilde{R}]
\end{align*}
is just equation (\ref{R-R-tilde}) from Proposition \ref{prop-bij-R}.
\end{remark}

In principle, for a fixed predictable process $R \in \calv$ it is possible to change the local martingale $\Theta \in \calm_{\loc}$ in order to obtain another ELMD. More precisely, we have the following result.

\begin{corollary}\label{cor-vary-Theta}
Suppose that the equivalent conditions from Theorem \ref{thm-special} are fulfilled. Let $T \in \calm_{\loc}$ be a local martingale with $T_0 = 0$ and $\Delta \Theta + \Delta T < 1$ such that $[M^i,T] \in \cala_{\loc}$ for each $i \in I$. We define the process $\hat{Z} := \hat{D} B^{-1}$, where $\hat{D} := \cale(-\Theta - T)$. Then the following statements are equivalent:
\begin{enumerate}
\item[(i)] $\hat{Z}$ is an ELMD for $\bbs$.

\item[(ii)] $\hat{Z}$ is an ELMD for $\bbs \cup \{ B \}$.

\item[(iii)] For each $i \in I$ we have up to an evanescent set
\begin{align*}
[M^i,T]^p = 0.
\end{align*}
\end{enumerate}
\end{corollary}

\begin{proof}
This is an immediate consequence of Theorem \ref{thm-special}.
\end{proof}

On the other hand, for a fixed $\Theta \in \calm_{\loc}$ it is not possible to change the predictable process $R \in \calv$ in order to obtain another ELMD. More precisely, we have the following result, which is in accordance with \cite[Prop. 7.9]{Platen-Tappe-FTAP}.

\begin{corollary}
Suppose that the equivalent conditions from Theorem \ref{thm-special} are fulfilled. Let $V \in \calv$ be a predictable process with $\Delta V > -1$, and let $\widetilde{V} \in \calv$ be the corresponding predictable process with $\Delta \widetilde{V} < 1$ according to Proposition \ref{prop-bij-R}. We define the process $\hat{Z} := D \hat{B}^{-1}$, where $\hat{B} := \cale(V)$. Then the following statements are equivalent:
\begin{enumerate}
\item[(i)] $\hat{Z}$ is an ELMD for $\bbs$.

\item[(ii)] $\hat{Z}$ is an ELMD for $\bbs \cup \{ \hat{B} \}$.

\item[(iii)] We have $B = \hat{B}$ up to an evanescent set.

\item[(iv)] We have $R = V$ up to an evanescent set.

\item[(v)] We have $\widetilde{R} = \widetilde{V}$ up to an evanescent set.
\end{enumerate}
\end{corollary}

\begin{proof}
(i) $\Leftrightarrow$ (ii): This follows from Theorem \ref{thm-special}.

\noindent(i) $\Leftrightarrow$ (v): Since $Z$ is an ELMD for $\bbs$, by Theorem \ref{thm-special} for each $i \in I$ we have (\ref{drift-R-Theta}) up to an evanescent set. Furthermore, by Theorem \ref{thm-special} the process $\hat{Z}$ is an ELMD for $\bbs$ if and only if for each $i \in I$ we have up to an evanescent set
\begin{align*}
A^i = \widetilde{V} + [X^i,Y]^p,
\end{align*}
which by (\ref{drift-R-Theta}) is equivalent to $\widetilde{R} = \widetilde{V}$ up to an evanescent set.

\noindent(iii) $\Leftrightarrow$ (iv): This equivalence is evident.

\noindent(iv) $\Leftrightarrow$ (v): This equivalence follows from Proposition \ref{prop-bij-R}.
\end{proof}

In the next result we determine the dynamics of the ELMN, provided it exists.

\begin{proposition}
Suppose that the equivalent conditions from Theorem \ref{thm-special} are fulfilled, and define the ELMN $\bar{Z} := Z^{-1}$. Let $\nu$ be the predictable compensator of the random measure $\mu^Y$. Then we have the representation
\begin{align*}
\bar{Z} = \cale \bigg( \widetilde{R} + \widetilde{\Theta} + \la \Theta^c,\Theta^c \ra + \frac{y^2}{1-y} * \mu^Y \bigg),
\end{align*}
and the following statements are equivalent:
\begin{enumerate}
\item[(i)] $\bar{Z}$ is a special semimartingale.

\item[(ii)] We have
\begin{align*}
\frac{y^2}{1-y} * \nu \in \cala_{\loc}^+.
\end{align*}
\end{enumerate}
If the previous conditions are fulfilled, then we have 
\begin{align*}
\bar{Z} = \cale( \bar{N} + \bar{B} ), 
\end{align*}
where the local martingale $\bar{N} \in \calm_{\loc}$ and the predictable process $\bar{B} \in \calv$ are given by
\begin{align*}
\bar{N} &= \widetilde{\Theta} + \frac{y^2}{1-y} * (\mu^Y - \nu),
\\ \bar{B} &= \widetilde{R} + \la \Theta^c,\Theta^c \ra + \frac{y^2}{1-y} * \nu.
\end{align*}
\end{proposition}

\begin{proof}
This is an immediate consequence of Proposition \ref{prop-inverse-pre}.
\end{proof}

We define the new market with discounted assets $\bbs B^{-1} := \{ S^i B^{-1} : i \in I \}$.

\begin{proposition}\label{prop-ELMM}
Suppose that $D \in \calm$ with $\bbp(D_{\infty} > 0)$, and let $\bbq \approx \bbp$ be the probability measure on $(\Omega,\calf_{\infty -})$ with density process $D$ relative to $\bbp$. Then the following statements are equivalent:
\begin{enumerate}
\item[(i)] $Z$ is an ELMD for $\bbs$.

\item[(ii)] $Z$ is an ELMD for $\bbs \cup \{ B \}$.

\item[(iii)] $\bbq$ is an ELMM for $\bbs B^{-1}$.
\end{enumerate}
\end{proposition}

\begin{proof}
(i) $\Leftrightarrow$ (ii): This equivalence follows from Theorem \ref{thm-special}.

\noindent(i) $\Leftrightarrow$ (iii): $Z$ is an ELMD for $\bbs$ if and only if $D$ is an ELMD for $\bbs B^{-1}$. By \cite[Lemma 4.7]{Platen-Tappe-FTAP} this is the case if and only if $\bbq$ is an ELMM for $\bbs B^{-1}$.
\end{proof}

\begin{remark}\label{rem-Girsanov}
Suppose that $D \in \calm$ with $\bbp(D_{\infty} > 0)$, and let $\bbq \approx \bbp$ be the probability measure on $(\Omega,\calf_{\infty -})$ with density process $D$ relative to $\bbp$. Let $i \in I$ be such that $[M^i,\Theta] \in \cala_{\loc}$. By Proposition \ref{prop-Girsanov} the process $X^i$ is also a special semimartingale under $\bbq$, and its canonical decomposition $X^i = (M^i)' + (A^i)'$ is given by
\begin{align*}
(M^i)' = M^i + [M^i,\Theta]^p \quad \text{and} \quad (A^i)' = A^i - [M^i,\Theta]^p,
\end{align*}
where the predictable compensators are computed under $\bbp$. By Yor's formula (see \cite[II.8.19]{Jacod-Shiryaev}) we have
\begin{align*}
S^i B^{-1} &= S_0^i \, \cale(X^i) \cale(-\widetilde{R}) = S_0^i \, \cale \big( X^i - \widetilde{R} - [X^i,\widetilde{R}] \big) = S_0^i \, \cale \big( (M^i)' + (B^i)' \big),
\end{align*}
where the predictable process $(B^i)' \in \calv$ is given by
\begin{align*}
(B^i)' = (A^i)' - \widetilde{R} - [A^i,\widetilde{R}].
\end{align*}
Therefore, $S^i B^{-1}$ is a $\bbq$-local martingale if and only if $(B^i)' = 0$, which is equivalent to (\ref{drift-R-Theta-2}), confirming Theorem \ref{thm-special} and Proposition \ref{prop-ELMM}.
\end{remark}

\begin{remark}\label{rem-decomposition}
Suppose that the equivalent conditions from Proposition \ref{prop-ELMM} are fulfilled. In view of Remark \ref{rem-Girsanov}, condition (\ref{drift-R-Theta-2}) reads
\begin{align*}
(A^i)' = \widetilde{R} + [A^i,\widetilde{R}],
\end{align*}
and we see that the process $\widetilde{R} + [A^i,\widetilde{R}]$ in (\ref{drift-R-Theta-2}) can be regarded as the \emph{locally risk-free return} of the asset $S^i$ if $M^i$ were zero, and that the process $\Theta$ in (\ref{drift-R-Theta-2}) can be regarded as a \emph{market price of risk}.
\end{remark}

\begin{remark}
Assume $I = \{ 1,\ldots,d \}$ for some $d \in \bbn$, and that the equivalent conditions from Proposition \ref{prop-ELMM} are fulfilled. Then $X = (X^1,\ldots,X^d)$ is an $\bbr^d$-valued special semimartingale. We denote by $(A,C,\nu)$ its characteristics. By Remark \ref{rem-Girsanov} and Proposition \ref{prop-Girsanov} the process $X$ is also a special martingale under $\bbq$, and its characteristics $(A',C',\nu')$ are given by
\begin{align*}
(A^i)' &= \widetilde{R} + [A^i,\widetilde{R}], \quad i=1,\ldots,d,
\\ C' &= C,
\\ \nu' &= \big( 1 - M_{\mu^X}^{\bbp}( \Delta \Theta \,|\, \widetilde{\calp} ) \big) \cdot \nu.
\end{align*}
\end{remark}

For the rest of this section, we assume that the special semimartingales $(X^i)_{i \in I}$ appearing in (\ref{deco-1}) and the special semimartingale $Y$ appearing in (\ref{deco-2}) are locally square-integrable and quasi-left-continuous. Then we have $M^i \in \calh_{\loc}^2$ for each $i \in I$, and we have $\widetilde{\Theta} \in \calh_{\loc}^2$, where $\calh_{\loc}^2$ denotes the space of all locally square-integrable martingales. Furthermore, by Lemma \ref{lemma-semi-qls}, for each $i \in I$ the local martingale $M^i$ is quasi-left-continuous and the process $A^i$ is continuous, and the local martingale $\widetilde{\Theta}$ is quasi-left-continuous and the process $\widetilde{R}$ is continuous. Using Propositions \ref{prop-bij-R} and \ref{prop-bij-Theta} we have $\Theta = \widetilde{\Theta}$, $R = \widetilde{R}$ and $B = \exp(R)$.

\begin{theorem}\label{thm-special-sq}
The following statements are equivalent:
\begin{enumerate}
\item[(i)] $Z$ is an ELMD for $\bbs$.

\item[(ii)] $Z$ is an ELMD for $\bbs \cup \{ B \}$.

\item[(iii)] For each $i \in I$ we have up to an evanescent set
\begin{align}\label{drift-R-Theta-sq}
A^i = R + \la M^i,\Theta \ra.
\end{align}
\end{enumerate}
If the previous conditions are fulfilled, then $\calp_{\sfi,1}^+(\bbs \cup \{ B \})$ satisfies NUPBR, NAA$_1$ and NA$_1$.
\end{theorem}

\begin{proof}
Taking into account Lemma \ref{lemma-qv-sq-qls}, this is an immediate consequence of Theorem \ref{thm-special}.
\end{proof}

Now we determine the dynamics of the ELMN in the present setting, provided it exists.

\begin{proposition}\label{prop-inverse-Z}
Suppose that the equivalent conditions from Theorem \ref{thm-special-sq} are fulfilled, and define the ELMN $\bar{Z} := Z^{-1}$. Let $\nu$ be the predictable compensator of the random measure $\mu^{\Theta}$. Then the following statements are equivalent:
\begin{enumerate}
\item[(i)] $\bar{Z}$ is a locally square-integrable semimartingale.

\item[(ii)] We have
\begin{align}\label{Theta-A-loc-qls}
\frac{\theta^2}{1 - \theta} * \nu, \bigg( \frac{\theta^2}{1 - \theta} \bigg)^2 * \nu \in \cala_{\loc}^+.
\end{align}

\item[(iii)] We have
\begin{align}\label{Theta-A-loc-qls-2}
\frac{\theta^2}{1 - \theta} * \nu, \bigg( \frac{\theta}{1 - \theta} \bigg)^2 * \nu \in \cala_{\loc}^+.
\end{align}

\item[(iv)] There exists a quasi-left-continuous local martingale $\bar{N} \in \calh_{\loc}^2$ with $\bar{N}_0 = 0$ such that
\begin{align}\label{N-properties}
\bar{N} - \Theta \in \calv, \quad \bar{N}^c = \Theta^c \quad \text{and} \quad \Delta \bar{N} = \frac{\Delta \Theta}{1 - \Delta \Theta}. 
\end{align}
\end{enumerate}
If the previous conditions are fulfilled, then the process $\bar{Z}$ admits the representation
\begin{align}\label{N-bar-qls}
\bar{Z} = \cale \big( R + \la \Theta,\bar{N} \ra + \bar{N} \big),
\end{align}
and the local martingale $\bar{N} \in \calm_{\loc}$ is given by
\begin{align*}
\bar{N} = \Theta + \frac{\theta^2}{1 - \theta} * (\mu^{\Theta} - \nu).
\end{align*}
\end{proposition}

\begin{proof}
This is an immediate consequence of Proposition \ref{prop-inverse-L2}.
\end{proof}

\section{Existence of equivalent local martingale deflators}\label{sec-existence}

In this section we treat the existence of ELMDs. We consider the framework of Section \ref{sec-ELMD} with $I = \{ 1,\ldots,d \}$ for some $d \in \bbn$; that is, we have finitely many assets. We introduce the $\bbr^d$-valued special semimartingale $X := (X^1,\ldots,X^d)$. As at the end of Section \ref{sec-ELMD}, we assume that $X$ is locally square-integrable and quasi-left-continuous. In view of Theorem \ref{thm-special}, we are interested in finding a continuous process $R \in \calv$ and a quasi-left-continuous local martingale $\Theta \in \calh_{\loc}^2$ with $\Theta_0 = 0$ and $\Delta \Theta < 1$ such that up to an evanescent set
\begin{align}\label{drift-for-ex}
A^i = R + \la M^i,\Theta \ra \quad \text{for all $i=1,\ldots,d$,}
\end{align}
because then
\begin{align*}
Z = \cale(-\Theta) \cale(R)^{-1} = \cale(-\Theta) \exp(R)^{-1}
\end{align*}
is an ELMD for $\bbs$. By Proposition \ref{prop-ex-int-char} there exist a continuous process $\Gamma \in \cala_{\loc}^+$ and modified integral characteristics $(a,c_{\mo},K)$ of $X$ with respect to $\Gamma$. Furthermore, by Proposition \ref{prop-int-char-eq} there exist integral characteristics $(a,c,K)$ and a purely discontinuous second integral characteristic $v$ of $X$ with respect to $\Gamma$, and we have
\begin{align}\label{c-mo-sum}
c_{\mo} = c + v \quad \text{$\Gamma$-a.e.} \quad \text{$\bbp$-a.e.}
\end{align}

\begin{proposition}\label{prop-ex-matrices}
The following statements are equivalent:
\begin{enumerate}
\item[(i)] There exist an optional $\bbr$-valued process $r$ and an optional $\bbs_+^{(d+1) \times (d+1)}$-valued process $\hat{c}_{\mo}$ such that $\hat{c}_{\mo}^{ij} = c_{\mo}^{ij}$ for all $i,j = 1,\ldots,d$, and we have
\begin{align}\label{eqn-ex-i}
\big( \hat{c}_{\mo}^{i,d+1} \big)_{i=1,\ldots,d} = a - r \bbI_{\bbr^d} \quad \text{$\Gamma$-a.e.} \quad \text{$\bbp$-a.e.}
\end{align}
\item[(ii)] There exist an optional $\bbr$-valued process $r$ and optional $\bbs_+^{(d+1) \times (d+1)}$-valued processes $\hat{c}$ and $\hat{v}$ such that $\hat{c}^{ij} = c^{ij}$ and $\hat{v}^{ij} = v^{ij}$ for all $i,j = 1,\ldots,d$, and we have
\begin{align}\label{eqn-ex-ii}
\big( \hat{c}^{i,d+1} \big)_{i=1,\ldots,d} + \big( \hat{v}^{i,d+1} \big)_{i=1,\ldots,d} = a - r \bbI_{\bbr^d} \quad \text{$\Gamma$-a.e.} \quad \text{$\bbp$-a.e.}
\end{align}
\item[(iii)] There exist an optional $\bbr$-valued process $r$ and an optional $\bbr^d$-valued process $x$ such that
\begin{align}\label{eqn-ex-iii}
c_{\mo} x = a - r \bbI_{\bbr^d} \quad \text{$\Gamma$-a.e.} \quad \text{$\bbp$-a.e.}
\end{align}

\item[(iv)] There exist an optional $\bbr$-valued process $r$ and optional $\bbr^d$-valued processes $x$ and $y$ such that
\begin{align}\label{eqn-ex-iv}
c x + v y = a - r \bbI_{\bbr^d} \quad \text{$\Gamma$-a.e.} \quad \text{$\bbp$-a.e.}
\end{align}
\end{enumerate}
\end{proposition}

\begin{proof}
Taking into account (\ref{c-mo-sum}), this is a consequence of Proposition \ref{prop-matrices-main}. Note that the corresponding processes can indeed be chosen to be optional, which follows from Lemma \ref{lemma-MP-measurable} and the additional statements from Lemma \ref{lemma-LGS} and Proposition \ref{prop-matrix}.
\end{proof}

\begin{remark}
Note that the equivalent conditions from Proposition \ref{prop-ex-matrices} are fulfilled if
\begin{align*}
c_{\mo} \in \bbs_{++}^{d \times d} \quad \text{$\Gamma$-a.e.} \quad \text{$\bbp$-a.e.}
\end{align*}
\end{remark}

The following results show that the existence of a continuous process $R \in \calv$ and a quasi-left-continuous local martingale $\Theta \in \calh_{\loc}^2$ satisfying (\ref{drift-for-ex}) is essentially equivalent to the existence of an optional $\bbr$-valued process $r$ and an optional $\bbr^d$-valued process $x$ satisfying (\ref{eqn-ex-iii}).

Let $R \in \calv$ be a continuous process, and let $\Theta \in \calh_{\loc}^2$ be a quasi-left-continuous local martingale with $\Theta_0 = 0$ and $\Delta \Theta < 1$. Denoting by $L_{\loc}^1(\Gamma)$ the space of all optional processes $r : \Omega \times \bbr_+ \to \bbr$ such that $|r| \bdot \Gamma \in \calv^+$, we assume there is an optional process $r \in L_{\loc}^1(\Gamma)$ such that $R = r \bdot \Gamma$, and that the $\bbr^{d+1}$-valued semimartingale $\hat{X} := (X,\Theta)$ admits modified integral characteristics $(\hat{a},\hat{c}_{\mo},\hat{K})$ with respect to $\Gamma$, where of course $\hat{a} = (a,0)$.

\begin{proposition}\label{prop-ex-1}
If condition (\ref{drift-for-ex}) is satisfied, then condition (\ref{eqn-ex-i}) is satisfied as well.
\end{proposition}

\begin{proof}
By (\ref{drift-for-ex}), for all $i=1,\ldots,d$ we have
\begin{align*}
\hat{c}_{\mo}^{i,d+1} \bdot \Gamma = \hat{C}_{\mo}^{i,d+1} = \la M^i,\Theta \ra = A^i - R = ( a^i - r ) \bdot \Gamma,
\end{align*}
showing that (\ref{eqn-ex-i}) is fulfilled.
\end{proof}

Now, we assume that the equivalent conditions from Proposition \ref{prop-ex-matrices} are fulfilled. By (\ref{c-mo-sum}), (\ref{eqn-ex-i}) and (\ref{eqn-ex-ii}) we may assume that
\begin{align*}
\hat{c}_{\mo} = \hat{c} + \hat{v} \quad \text{$\Gamma$-a.e.} \quad \text{$\bbp$-a.e.}
\end{align*}

\begin{lemma}
There is a transition kernel $\hat{K}$ from $(\Omega \times \bbr_+, \calo)$ into $(\bbr^{d+1},\calb(\bbr^{d+1}))$ such that on $\Omega \times \bbr_+$ we have
\begin{align}\label{kernel-ext-0}
\hat{K}(\{ 0 \}) = 0 \quad \text{and} \quad \int_{\bbr^{d+1}} |\hat{x}|^2 \hat{K}(d \hat{x}) < \infty,
\end{align}
for every nonnegative, measurable function $f : \bbr^d \to \bbr_+$ we have
\begin{align}\label{kernel-ext-1}
\int_{\bbr^d} f(x) K(dx) = \int_{\bbr^{d+1}} f(x) \hat{K}(d \hat{x}),
\end{align}
and for all $i,j = 1,\ldots,d+1$ with $i \leq d$ or $j \leq d$ we have
\begin{align}\label{kernel-ext-2}
\hat{v}^{ij} = \int_{\bbr^{d+1}} \hat{x}^i \hat{x}^j \hat{K}(d \hat{x}).
\end{align}
\end{lemma}

\begin{proof}
This is a consequence of Lemma \ref{lemma-measure} and Fubini's theorem for transition kernels, where we note Lemma \ref{lemma-MP-measurable}.
\end{proof}

By adjusting $\hat{v}^{d+1,d+1}$ if necessary, we even have (\ref{kernel-ext-2}) for all $i,j = 1,\ldots,d+1$. Note that this does not affect equation (\ref{eqn-ex-ii}). We assume that $r \in L_{\loc}^1(\Gamma)$ and define the continuous process $R \in \calv$ as $R := r \bdot \Gamma$. Furthermore, we assume there exists a quasi-left-continuous local martingale $\Theta \in \calh_{\loc}^2$ with $\Theta_0 = 0$ and $\Delta \Theta < 1$ such that the $\bbr^{d+1}$-valued semimartingale $\hat{X} = (X,\Theta)$ has the modified integral characteristics $(\hat{a},\hat{c}_{\mo},\hat{K})$ with respect to $\Gamma$, where $\hat{a} = (a,0)$. Note that the latter condition is related to the martingale problem (see \cite[Sec. III.2]{Jacod-Shiryaev}), which can be solved in many situations.

\begin{proposition}\label{prop-ex-2}
Under the previous assumptions, condition (\ref{drift-for-ex}) is fulfilled.
\end{proposition}

\begin{proof}
Using (\ref{eqn-ex-i}), for all $i=1,\ldots,d$ we have up to an evanescent set
\begin{align*}
A^i - R = (a^i - r) \bdot \Gamma = \hat{c}_{\mo}^{i,d+1} \bdot \Gamma = \hat{C}_{\mo}^{i,d+1} = \la M^i,\Theta \ra,
\end{align*}
showing that condition (\ref{drift-for-ex}) is fulfilled.
\end{proof}

Summing up, Propositions \ref{prop-ex-1} and \ref{prop-ex-2} show that the existence of an ELMD $Z$, which is a multiplicative special semimartingale, is essentially, up to a solution to the martingale problem, equivalent to the existence of optional processes $r$ and $x$ satisfying (\ref{eqn-ex-iii}).

\begin{example}
Assume that
\begin{align*}
c_{\mo} =
\left(
\begin{array}{cc}
1 & 1
\\ 1 & 1
\end{array}
\right) \quad \text{$\Gamma$-a.e.} \quad \text{$\bbp$-a.e.}
\end{align*}
Then equation (\ref{eqn-ex-iii}) has a solution if and only if 
\begin{align*}
a \in \lin \{ \bbI_{\bbr^2} \} \quad \text{$\Gamma$-a.e.} \quad \text{$\bbp$-a.e.}
\end{align*}
where $\lin \{ \bbI_{\bbr^2} \}$ denotes the linear spaces generated by the vector $\bbI_{\bbr^2} = (1,1)$.
\end{example}

\section{Jump-diffusion models}\label{sec-jd}

In this section we study the existence of ELMDs for jump-diffusion models. Let $\lambda$ be the Lebesgue measure on $(\bbr_+,\calb(\bbr_+))$, and let $W$ be an $\bbr^m$-valued standard Wiener process for some $m \in \bbn$. Furthermore, let $\frp$ be a homogeneous Poisson random measure on some mark space $(E,\cale)$, which we assume to be a Blackwell space. Then its compensator is of the form $\frq = \lambda \otimes F$ with some $\sigma$-finite measure $F$ on the mark space $(E,\cale)$. Let $L_{\loc}^1(\lambda)$ be the space of all optional processes $a : \Omega \times \bbr_+ \to \bbr$ such that $|a| \bdot \lambda \in \calv^+$, let $L_{\loc}^2(W)$ be the space of all optional processes $\sigma : \Omega \times \bbr_+ \to \bbr^m$ such that $\| \sigma \|_{\bbr^m}^2 \bdot \lambda \in \calv^+$, and let $L_{\loc}^2(\frp)$ be the space of all optional processes $\gamma : \Omega \times \bbr_+ \times E \to \bbr$ such that $|\gamma|^2 * \frq \in \calv^+$.

\begin{remark}
In view of our upcoming results such as condition (\ref{drift-jd}) below, we emphasize that we may assume that the processes from $L_{\loc}^2(W)$ and $L_{\loc}^2(\frp)$ are optional. For example, for each $\sigma \in L_{\loc}^2(W)$ we have
\begin{align*}
\sigma = {}^p \sigma \quad \text{$\lambda$-a.e.} \quad \text{$\bbp$-a.e.}
\end{align*}
where ${}^p \sigma$ denotes the predictable projection of $\sigma$. Indeed, by Fubini's theorem we have
\begin{align*}
\bbe \bigg[ \int_0^t \sigma_s ds \bigg] &= \int_0^t \bbe[\sigma_s] ds = \int_0^t \bbe[\bbe[\sigma_s \,|\, \calf_{s-}]] ds
\\ &= \int_0^t \bbe[({}^p \sigma)_s] ds = \bbe \bigg[ \int_0^t ({}^p \sigma)_s ds \bigg] \quad \text{for each $t \in \bbr_+$.} 
\end{align*}
\end{remark}

As in the previous sections, we suppose that the market is given by
\begin{align*}
\bbs = \{ S^i : i \in I \}, 
\end{align*}
where for each $i \in I$ the asset $S^i$ is given by a stochastic exponential (\ref{stock-i}). Here we assume that for each $i \in I$ the semimartingale $X^i$ in (\ref{stock-i}) is given by
\begin{align*}
X^i = a^i \bdot \lambda + \sigma^i \bdot W + \gamma^i * (\frp - \frq)
\end{align*}
with $a^i \in L_{\loc}^1(\lambda)$, $\sigma^i \in L_{\loc}^2(W)$ and $\gamma^i \in L_{\loc}^2(\frp)$ such that $\gamma^i > -1$. Then for each $i \in I$ the semimartingale $X^i$ is locally square-integrable and quasi-left-continuous, and hence we are in the framework considered at the end of Section \ref{sec-ELMD}. In order to look for ELMDs which are multiplicative special semimartingales, we consider a multiplicative special semimartingale $Z = D B^{-1}$, where
\begin{align}\label{Z-prod-jd}
D = \cale \big( -\theta \bdot W - \psi * (\frp - \frq) \big) \quad \text{and} \quad B = \cale(r \bdot \lambda) = \exp(r \bdot \lambda)
\end{align}
with $\theta \in L_{\loc}^2(W)$ and $\psi \in L_{\loc}^2(\frp)$ such that $\psi < 1$, and an optional process $r \in L_{\loc}^1(\lambda)$. Note that for each $i \in I$ the process $\gamma^i$ can be considered as an $L^2(F)$-valued process, and analogously $\psi$ can be considered as an $L^2(F)$-valued process.

\begin{theorem}\label{thm-jd}
The following statements are equivalent:
\begin{enumerate}
\item[(i)] $Z$ is an ELMD for $\bbs$.

\item[(ii)] $Z$ is an ELMD for $\bbs \cup \{ B \}$.

\item[(iii)] For each $i \in I$ we have
\begin{align}\label{drift-jd}
\langle \sigma^i,\theta \rangle_{\bbr^m} + \langle \gamma^i, \psi \rangle_{L^2(F)} = a^i - r \quad \text{$\lambda$-a.e.} \quad \text{$\bbp$-a.e.}
\end{align}
\end{enumerate}
If the previous conditions are fulfilled, then $\calp_{\sfi,1}^+(\bbs \cup \{ B \})$ satisfies NUPBR, NAA$_1$ and NA$_1$.
\end{theorem}

\begin{proof}
This is an immediate consequence of Theorem \ref{thm-special-sq}.
\end{proof}

In the upcoming result we consider the situation where the deflator admits a measure change.

\begin{proposition}
Suppose that $D \in \calm$ with $\bbp(D_{\infty} > 0)$, and let $\bbq \approx \bbp$ be the probability measure on $(\Omega,\calf_{\infty -})$ with density process $D$ relative to $\bbp$. Then the process $W' := W + \theta \bdot \lambda$ is an $\bbr^m$-valued $\bbq$-standard Wiener process, the random measure $\frp$ is a $\bbq$-integer valued random measure with compensator given by $\frq' = (1-\psi) \cdot \frq$, and the following statements are equivalent:
\begin{enumerate}
\item[(i)] $Z$ is an ELMD for $\bbs$.

\item[(ii)] $Z$ is an ELMD for $\bbs \cup \{ B \}$.

\item[(iii)] $\bbq$ is an ELMM for $\bbs B^{-1}$.

\item[(iv)] For each $i \in I$ we have up to an evanescent set
\begin{align*}
X^i = r \bdot \lambda + \sigma^i \bdot W' + \gamma^i * (\frp - \frq'),
\end{align*}
\end{enumerate}
\end{proposition}

\begin{proof}
The statements about $W'$ and $\frp$ follow from Proposition \ref{prop-Girsanov} combined with L\'{e}vy's theorem (see \cite[Thm. II.4.4]{Jacod-Shiryaev}).

\noindent(i) $\Leftrightarrow$ (ii) $\Leftrightarrow$ (iii): This is a consequence of Proposition \ref{prop-ELMM}.

\noindent(iii) $\Leftrightarrow$ (iv): For each $i \in I$ we have
\begin{align*}
&r \bdot \lambda + \sigma^i \bdot W' + \gamma^i * (\frp - \frq')
\\ &= r \bdot \lambda + \sigma^i \bdot W + \sigma^i \bdot (\theta \bdot \lambda) + \gamma^i * (\frp - \frq) + \gamma^i * (\frq - \frq')
\\ &= \big( r + \langle \sigma^i,\theta \rangle_{\bbr^m} + \langle \gamma^i, \psi \rangle_{L^2(F)} \big) \bdot \lambda + \sigma^i \bdot W + \gamma^i * (\frp - \frq),
\end{align*}
and hence, this equivalence follows from Theorem \ref{thm-jd}.
\end{proof}

In the next result we investigate when the corresponding ELMN is locally square-integrable, and derive its dynamics in this case.

\begin{proposition}
Suppose that the equivalent conditions from Theorem \ref{thm-jd} are fulfilled, and define the ELMN $\bar{Z} := Z^{-1}$. Then the following statements are equivalent:
\begin{enumerate}
\item[(i)] $\bar{Z}$ is a locally square-integrable semimartingale.

\item[(ii)] We have
\begin{align*}
\frac{\psi}{\sqrt{1 - \psi}}, \frac{\psi^2}{1 - \psi} \in L_{\loc}^2(\frp).
\end{align*}

\item[(iii)] We have
\begin{align*}
\frac{\psi}{\sqrt{1 - \psi}}, \frac{\psi}{1 - \psi} \in L_{\loc}^2(\frp).
\end{align*}
\end{enumerate}
If the previous conditions are fulfilled, then the process $\bar{Z}$ admits the representation
\begin{align}\label{ELMN-dynamics}
\bar{Z} = \cale \bigg( \bigg( r + \| \theta \|_{\bbr^m}^2 + \bigg\langle \psi, \frac{\psi}{1 - \psi} \bigg\rangle_{L^2(F)} \bigg) \bdot \lambda + \theta \bdot W + \frac{\psi}{1 - \psi} * (\frp - \frq) \bigg).
\end{align}
\end{proposition}

\begin{proof}
This is an immediate consequence of Proposition \ref{prop-inverse-Z}.
\end{proof}

\begin{remark}
The representation (\ref{ELMN-dynamics}) has been derived in earlier works as the structure of a growth optimal portfolio; see, for example, the articles \cite{Christensen-Platen}, \cite{B-N-Platen} and \cite{Platen-Tappe}.
\end{remark}

Now, we consider the situation with finitely many assets; that is $I = \{ 1,\ldots,d \}$ for some $d \in \bbn$. We define the $\bbr^d$-valued semimartingale $X := (X^1,\ldots,X^d)$ and the $\bbr^d$-valued process $a := (a^1,\ldots,a^d)$. Furthermore, we define the optional $\bbs_+^{d \times d}$-valued process as
\begin{align*}
c^{ij} = \la \sigma^i,\sigma^j \ra_{\bbr^m} \quad \text{for all $i,j = 1,\ldots,d$.}
\end{align*}
We define the transition kernel $K$ from $(\Omega \times \bbr_+, \calo)$ into $(\bbr^d,\calb(\bbr^d))$ as the image measure $K := F \circ \gamma$, and we define the $\bbs_+^{d \times d}$-valued process $v$ as
\begin{align*}
v^{ij} := \la \gamma^i,\gamma^j \ra_{L^2(F)} \quad \text{for all $i,j = 1,\ldots,d$.}
\end{align*}
Furthermore, we define the $\bbs_+^{d \times d}$-valued process
\begin{align*}
c_{\mo} := c + v.
\end{align*}
The following obvious auxiliary result shows that we are in the framework of Section~\ref{sec-existence}.

\begin{lemma}
The following statements are true:
\begin{enumerate}
\item The triplet $(a,c,K)$ consists of integral characteristics of $X$ with respect to $\lambda$.

\item The triplet $(a,c_{\mo},K)$ consists of modified integral characteristics of $X$ with respect to $\lambda$.

\item The process $v$ is a purely discontinuous second integral characteristic of $X$ with respect to $\lambda$.
\end{enumerate}
\end{lemma}

By identification, we may regard $\sigma$ as the $L(\bbr^m,\bbr^d)$-valued process
\begin{align}\label{ident-1}
\sigma = \big( \la \sigma^i,\cdot \ra_{\bbr^m} \big)_{i=1,\ldots,d},
\end{align}
where we note that $L(\bbr^m,\bbr^d) \cong \bbr^{d \times m}$. 

Similarly, we may regard $\gamma$ as the $L(L^2(F),\bbr^d)$-valued process
\begin{align}\label{ident-2}
\gamma = \big( \la \gamma^i,\cdot \ra_{L^2(F)} \big)_{i=1,\ldots,d}.
\end{align}

\begin{proposition}\label{prop-jd}
The following statements are equivalent:
\begin{enumerate}
\item[(i)] There exist an optional $\bbr$-valued process $r$, an optional $\bbr^m$-valued processes $\theta$ and an optional $L^2(F)$-valued process $\psi$ such that for each $i=1,\ldots,d$ we have (\ref{drift-jd}).

\item[(ii)] There exist an optional $\bbr$-valued process $r$, an optional $\bbr^m$-valued process $\theta$ and an optional $L^2(F)$-valued process $\psi$ such that
\begin{align}\label{jd-ex-1}
\sigma \theta + \gamma \psi = a - r \bbI_{\bbr^d} \quad \text{$\lambda$-a.e.} \quad \text{$\bbp$-a.e.}
\end{align}

\item[(iii)] There exist an optional $\bbr$-valued process $r$ and optional $\bbr^d$-valued processes $x$ and $y$ such that
\begin{align}
c x + vy = a - r \bbI_{\bbr^d} \quad \text{$\lambda$-a.e.} \quad \text{$\bbp$-a.e.}
\end{align}

\item[(iv)] There exist an optional $\bbr$-valued process $r$ and a optional $\bbr^d$-valued process $x$ such that
\begin{align}\label{eq-c-mod-jd}
c_{\mo} x = a - r \bbI_{\bbr^d} \quad \text{$\lambda$-a.e.} \quad \text{$\bbp$-a.e.}
\end{align}

\end{enumerate}
\end{proposition}

\begin{proof}
(i) $\Leftrightarrow$ (ii): Using the identifications (\ref{ident-1}) and (\ref{ident-2}), this equivalence is obvious.

\noindent(ii) $\Leftrightarrow$ (iii): Let $T \in L(\bbr^m \oplus_2 L^2(F),\bbr^d)$ be the continuous linear operator given by
\begin{align*}
T(\theta,\psi) := \sigma \theta + \gamma \psi, \quad \theta \in \bbr^m \text{ and } \psi \in L^2(F).
\end{align*}
Then the linear equation (\ref{jd-ex-1}) can equivalently be written as
\begin{align*}
T(\theta,\psi) = a - r \bbI_{\bbr^d} \quad \text{$\lambda$-a.e.} \quad \text{$\bbp$-a.e.}
\end{align*}
and hence this equivalence follows from Lemmas \ref{lemma-solution-kernel}--\ref{lemma-adjungiert}.

\noindent(iii) $\Leftrightarrow$ (iv): This equivalence follows from Proposition \ref{prop-ex-matrices}.
\end{proof}

Note that Theorem \ref{thm-jd} and Proposition \ref{prop-jd} have the following consequences. If an ELMD $Z = D B^{-1}$ of the form (\ref{Z-prod-jd}) exists, then the linear equation (\ref{eq-c-mod-jd}) has a solution $(r,x)$. Conversely, if the linear equation (\ref{eq-c-mod-jd}) has a solution, then -- subject to the conditions $r \in L_{\loc}^1(\lambda)$, $\theta \in L_{\loc}^2(W)$ and $\psi \in L_{\loc}^2(\frp)$ with $\psi < 1$ -- an ELMD $Z = D B^{-1}$ of the form (\ref{Z-prod-jd}) exists as well. This is in accordance with the findings of Section \ref{sec-existence}, but here we do not have to deal with the martingale problem.

\section{Conclusion}\label{sec-conclusion}

In this paper we have provided a systematic investigation on the existence of ELMDs, which are multiplicative special semimartingales, for a given market $\bbs$. There are connected questions which give rise to future research projects. One issue is the tradeability of the deflator $Z = D B^{-1}$; that is, whether the corresponding ELMN $\bar{Z} = Z^{-1}$ can be realized as a self-financing portfolio constructed in the extended market $\bbs \cup \{ B \}$. Even if it cannot be replicated, it arises the question how a central bank can approximate the ELMN $\bar{Z}$, which gives rise to diversification.

\begin{appendix}

\section{Semimartingales}\label{sec-processes}

In this appendix we provide the required results about semimartingales.

\begin{definition}
An adapted c\`{a}dl\`{a}g process $X$ is called \emph{quasi-left-continuous} if $\Delta X_T = 0$ almost surely on $\{ T < \infty \}$ for every predictable time $T$.
\end{definition}

\begin{lemma}\label{lemma-disjoint-jumps}
Let $X$ and $Y$ be two adapted c\`{a}dl\`{a}g processes such that $X$ is predictable and $Y$ is quasi-left-continuous. Then we have $\{ \Delta X \neq 0 \} \cap \{ \Delta Y \neq 0 \} = \emptyset$ up to an evanescent set.
\end{lemma}

\begin{proof}
Since $X$ is predictable, by \cite[Prop. I.2.24]{Jacod-Shiryaev} there exists an exhausting sequence $(S_n)_{n \in \bbn}$ of predictable times such that
\begin{align*}
\{ \Delta X \neq 0 \} = \bigcup_{n \in \bbn} \IL S_n \IR.
\end{align*}
Since $Y$ is quasi-left-continuous, by \cite[Prop. I.2.26]{Jacod-Shiryaev} there exists an exhausting sequence $(T_m)_{m \in \bbn}$ of totally inaccessible stopping times such that
\begin{align*}
\{ \Delta Y \neq 0 \} = \bigcup_{m \in \bbn} \IL T_m \IR.
\end{align*}
Therefore, we obtain up to an evanescent set
\begin{align*}
\{ \Delta X \neq 0 \} \cap \{ \Delta Y \neq 0 \} &= \bigg( \bigcup_{n \in \bbn} \IL S_n \IR \bigg) \cap \bigg( \bigcup_{m \in \bbn} \IL T_m \IR \bigg)
\\ &= \bigcup_{n,m \in \bbn} \big( \IL S_n \IR \cap \IL T_m \IR \big) = \emptyset,
\end{align*}
completing the proof.
\end{proof}

\begin{definition}
A semimartingale $X$ is called a \emph{special semimartingale} if there exists a semimartingale decomposition $X = X_0 + M + A$ such that $A$ is predictable.
\end{definition}

Let $X$ be a special semimartingale. Then the decomposition $X = X_0 + M + A$ with a predictable process $A \in \calv$ is unique up to an evanescent set (see \cite[I.3.16]{Jacod-Shiryaev}) and we call $X = X_0 + M + A$ the \emph{canonical decomposition} of $X$.

\begin{definition}
A semimartingale $X$ is called \emph{locally square-integrable} if it is a special semimartingale with canonical decomposition $X = X_0 + M + A$ satisfying $M \in \calh_{\loc}^2$.
\end{definition}

\begin{lemma}\label{lemma-semi-qls}
For a special semimartingale $X = X_0 + M + A$ the following statements are equivalent:
\begin{enumerate}
\item[(i)] $X$ is quasi-left-continuous.

\item[(ii)] $M$ is quasi-left-continuous and $A$ is continuous.
\end{enumerate}
\end{lemma}

\begin{proof}
(i) $\Rightarrow$ (ii): By \cite[Cor. I.2.31]{Jacod-Shiryaev} we have ${}^p(\Delta X) = \Delta A$. Since $X$ is quasi-left-continuous, we have
\begin{align*}
\Delta A_T = \bbe[ \Delta X_T \,|\, \calf_{T-} ] = 0
\end{align*}
almost surely on $\{ T < \infty \}$ for every predictable time $T$. Therefore, $A$ is quasi-left-continuous, and hence $M$ is quasi-left-continuous as well. Since $A$ is also predictable, by \cite[Prop. I.2.18.b]{Jacod-Shiryaev} we deduce that $A$ is continuous.

\noindent (ii) $\Rightarrow$ (i): This implication is obvious.
\end{proof}

\section{The quadratic variation}

In this appendix we provide the required results about the quadratic variation of semimartingales.

\begin{lemma}\label{lemma-cov-R-Theta}
Let $M \in \calm_{\loc}$ be a local martingale, and let $A \in \calv$ be a predictable process. Then the following statements are true:
\begin{enumerate}
\item We have $[M,A] \in \calm_{\loc}$.

\item We have $[M,A] \in \cala_{\loc}$ and $[M,A]^p = 0$
\end{enumerate}
\end{lemma}

\begin{proof}
The first statement follows from \cite[Prop. I.4.49.c]{Jacod-Shiryaev}, and the second statement is a consequence of \cite[Lemma I.3.11 and I.3.22]{Jacod-Shiryaev}.
\end{proof}

\begin{lemma}\label{lemma-cov-3-fold}
Let $M,N \in \calm_{\loc}$ be local martingales, and let $A \in \calv$ be a predictable process. Then the following statements are true:
\begin{enumerate}
\item We have $[M,[N,A]] \in \calm_{\loc}$.

\item We have $[M,[N,A]] \in \cala_{\loc}$ and $[M,[N,A]]^p = 0$.
\end{enumerate}
\end{lemma}

\begin{proof}
By \cite[Prop. 4.49.c]{Jacod-Shiryaev} we have $[N,A] \in \calm_{\loc}$. Therefore, by \cite[Prop. 4.49.a]{Jacod-Shiryaev} we obtain
\begin{align*}
[M,[N,A]] = \Delta M \bdot [N,A] \in \calm_{\loc}.
\end{align*}
The second statement is a consequence of \cite[Lemma I.3.11 and I.3.22]{Jacod-Shiryaev}.
\end{proof}

\begin{lemma}\label{lemma-cov-equivalence}
Let $X$ and $Y$ be two special semimartingales with canonical decompositions $X = M+A$ and $Y = N+B$. Then the following statements are equivalent:
\begin{enumerate}
\item[(i)] We have $[X,Y] \in \cala_{\loc}$.

\item[(ii)] We have $[M,N] \in \cala_{\loc}$.

\item[(iii)] We have $[M,\widetilde{N}] \in \cala_{\loc}$, where $\widetilde{N} := N-[N,B]$.
\end{enumerate}
In either case, we have
\begin{align}\label{qv-X-Y}
[X,Y]^p = [A,B] + [M,N]^p = [A,B] + [M,\widetilde{N}]^p,
\end{align}
and the quadratic variation $[A,B]$ is given by
\begin{align}\label{qv-A-B}
[A,B] = \sum_{s \leq \bullet} \Delta A_s \Delta B_s.
\end{align}
\end{lemma}

\begin{proof}
Note the decomposition
\begin{align*}
[X,Y] = [M,N] + [M,A] + [N,B] + [A,B].
\end{align*}
Furthermore, by \cite[Thm. I.4.52]{Jacod-Shiryaev} we have (\ref{qv-A-B}). Therefore, the quadratic variation $[A,B]$ is predictable, and hence, by \cite[Lemma I.3.10]{Jacod-Shiryaev} we have $[A,B] \in \cala_{\loc}$ with $[A,B]^p = [A,B]$. Consequently, the equivalences (i) $\Leftrightarrow$ (ii) $\Leftrightarrow$ (iii) and the formula (\ref{qv-X-Y}) follow from Lemmas \ref{lemma-cov-R-Theta} and \ref{lemma-cov-3-fold}.
\end{proof}

\begin{lemma}\label{lemma-qv-sq-qls}
Let $X$ and $Y$ be two locally square-integrable, quasi-left-continuous semimartingales with canonical decompositions $X = M+A$ and $Y = N+B$. Then we have $[X,Y] \in \cala_{\loc}$ and 
\begin{align*}
[X,Y]^p = \la M,N \ra = \la M^c,N^c \ra + \bigg[ \sum_{s \leq \bullet} \Delta M_s \Delta N_s \bigg]^p.
\end{align*}
\end{lemma}

\begin{proof}
By Lemma \ref{lemma-semi-qls} the processes $A$ and $B$ are continuous. Therefore, the statement is an immediate consequence of Lemma \ref{lemma-cov-equivalence} and \cite[Prop. I.4.50.b and Thm. I.4.52]{Jacod-Shiryaev}.
\end{proof}

\section{The stochastic exponential}

In this appendix we provide the required results about the stochastic exponential of a semimartingale.

\begin{lemma}\label{lemma-stoch-exp-properties}
Let $X$ be a semimartingale with $X_0 = 0$ and $\Delta X > -1$, and set $Z := \cale(X)$. Then the following statements are true:
\begin{enumerate}
\item $X$ is a special semimartingale if and only if $Z$ is a special semimartingale.

\item $X$ is a locally square-integrable semimartingale if and only if $Z$ is a locally square-integrable semimartingale.

\item $X$ is quasi-left-continuous if and only if $Z$ is quasi-left-continuous.
\end{enumerate}
\end{lemma}

\begin{proof}
Noting that $Z = 1 + Z_- \bdot X$ and $X = (Z_-)^{-1} \bdot Z$, the proof is immediate.
\end{proof}

\begin{definition}
Let $S$ be a semimartingale with $S,S_- > 0$.
\begin{enumerate}
\item[(a)] $S$ is called \emph{inversely special} if $S^{-1}$ is a special semimartingale.

\item[(b)] $S$ is called \emph{inversely locally square-integrable} if $S^{-1}$ is a locally square-integrable semimartingale.
\end{enumerate}
\end{definition}

In the definition (\ref{def-sum-jumps}) below we follow the convention from \cite[II.1.5]{Jacod-Shiryaev} to put the integral equal to $+\infty$ if it diverges.

\begin{lemma}\label{lemma-jumps-compensator}
Let $D \subset \bbr^d$ be a subset containing zero, and let $X$ be a $D$-valued c\`{a}dl\`{a}g, adapted process. Furthermore, let $\varphi : D \to \bbr$ be a measurable mapping with $\varphi(0) = 0$, and set
\begin{align}\label{def-sum-jumps}
A := \varphi(x) * \mu^X = \sum_{s \leq \bullet} \varphi(\Delta X_s).
\end{align}
Let $\nu$ be the predictable compensator of the random measure $\mu^X$. Then the following statements are equivalent:
\begin{enumerate}
\item[(i)] We have $A \in \cala_{\loc}$.

\item[(ii)] We have $\varphi(x) * \nu \in \cala_{\loc}$.
\end{enumerate}
In either case, the following statements are true:
\begin{enumerate}
\item We have $A^p = \varphi(x) * \nu$.

\item We have $A - A^p = \varphi(x) * (\mu^X - \nu)$.

\item We have $\Delta (A^p) = {}^p [ \varphi(\Delta X) ]$.

\item If $X$ is quasi-left-continuous, then we have $\Delta (A^p) = 0$.
\end{enumerate}
\end{lemma}

\begin{proof}
We have $A \in \cala_{\loc}$ if and only if $|\varphi(x)| * \mu^X \in \cala_{\loc}^+$, and we have $\varphi(x) * \nu \in \cala_{\loc}$ if and only if $|\varphi(x)| * \nu \in \cala_{\loc}^+$. Hence, the equivalence (i) $\Leftrightarrow$ (ii) follows from \cite[Thm. II.1.8.i]{Jacod-Shiryaev}. Now assume that $A \in \cala_{\loc}$. By \cite[Thm. II.1.8.ii]{Jacod-Shiryaev} we have $A^p = \varphi(x) * \nu$, and by \cite[Prop. II.1.28]{Jacod-Shiryaev} we have $A - A^p = \varphi(x) * (\mu^X - \nu)$. Furthermore, by \cite[I.3.21]{Jacod-Shiryaev} we have $\Delta(A^p) = {}^p (\Delta A)$. Since $\varphi(0) = 0$, we have $\Delta A = \varphi(\Delta X)$, and hence we obtain $\Delta (A^p) = {}^p [ \varphi(\Delta X) ]$. Now assume that $X$ is quasi-left-continuous. Then we have
\begin{align*}
\Delta A_T = \varphi(\Delta X_T) = 0
\end{align*}
almost surely on the set $\{ T < \infty \}$ for every predictable time $T$. Hence, by the definition of the predictable projection (see \cite[Thm. I.2.28.a]{Jacod-Shiryaev}) and \cite[Prop. I.2.18.b]{Jacod-Shiryaev} we deduce that $\Delta (A^p) = {}^p (\Delta A) = 0$.
\end{proof}

\begin{proposition}\label{prop-inverse-pre}
Let $X = M + A$ be a special semimartingale with $\Delta X > -1$, denote by $\nu$ the predictable compensator of $\mu^X$, and set $S := \cale(X)$. Then we have
\begin{align}\label{S-inverse-general}
S^{-1} = \cale \bigg( - X + \langle X^c,X^c \rangle + \frac{x^2}{1 + x} * \mu^X \bigg),
\end{align}
and the following statements are equivalent:
\begin{enumerate}
\item[(i)] $S$ is inversely special.

\item[(ii)] We have
\begin{align}\label{inv-spec-cond}
\frac{x^2}{1 + x} * \nu \in \cala_{\loc}^+.
\end{align}
\end{enumerate}
If the previous conditions are fulfilled, then we have
\begin{align}\label{repr-inverse}
S^{-1} = \cale(N+B),
\end{align}
where the local martingale $N \in \calm_{\loc}$ and the predictable process $B \in \calv$ are given by
\begin{align*}
N &= -M + \frac{x^2}{1+x} * (\mu^X - \nu),
\\ B &= -A + \la M^c,M^c \ra + \frac{x^2}{1+x} * \nu.
\end{align*}
\end{proposition}

\begin{proof}
The identity (\ref{S-inverse-general}) follows from \cite[Lemma 3.4]{Karatzas-Kardaras}. Noting (\ref{S-inverse-general}), by \cite[Prop. I.4.23]{Jacod-Shiryaev} and Lemma \ref{lemma-stoch-exp-properties} the semimartingale $S^{-1}$ is a special semimartingale if and only we have 
\begin{align*}
\frac{x^2}{1 + x} * \mu^X \in \cala_{\loc}^+,
\end{align*}
which is equivalent to (\ref{inv-spec-cond}) according to Lemma \ref{lemma-jumps-compensator}. Now, assume that (\ref{inv-spec-cond}) is fulfilled. Using Lemma \ref{lemma-jumps-compensator}, we arrive at the representation (\ref{repr-inverse}).
\end{proof}

\begin{proposition}\label{prop-inverse-L2}
Let $X = M + A$ be a locally square-integrable and quasi-left-continuous semimartingale with $\Delta M > -1$ (or equivalently $\Delta X > -1$), denote by $\nu$ the predictable compensator of $\mu^M$, and set $S := \cale(X)$. Then the following statements are equivalent:
\begin{enumerate}
\item[(i)] $S$ is inversely locally square-integrable.

\item[(ii)] We have
\begin{align}\label{spec-both-cond}
\frac{x^2}{1+x} * \nu, \bigg( \frac{x^2}{1+x} \bigg)^2 * \nu \in \cala_{\loc}^+.
\end{align}

\item[(iii)] We have
\begin{align}
\frac{x^2}{1+x} * \nu, \bigg( \frac{x}{1+x} \bigg)^2 * \nu \in \cala_{\loc}^+.
\end{align}
\item[(iv)] There exists a quasi-left-continuous local martingale $N \in \calh_{\loc}^2$ with $N_0 = 0$ such that
\begin{align}\label{N-martingale}
M+N \in \calv, \quad N^c = -M^c \quad \text{and} \quad \Delta N = -\frac{\Delta M}{1 + \Delta M}.
\end{align}
\end{enumerate}
If the previous conditions are fulfilled, then we have the representation
\begin{align}\label{repr-S-2}
S^{-1} = \cale \big( -A - \la M,N \ra + N \big),
\end{align}
and the local martingale $N \in \calm_{\loc}$ is given by
\begin{align}\label{def-N-appendix}
N = -M + \frac{x^2}{1+x} * (\mu^M - \nu).
\end{align}
\end{proposition}

\begin{proof}
(i) $\Leftrightarrow$ (ii): By Proposition \ref{prop-inverse-pre} the process $S$ is inversely special if and only if we have (\ref{inv-spec-cond}), and in this case we have the representation 
\begin{align}\label{inverse-proof-2}
S^{-1} = \cale(N+B),
\end{align}
where the local martingale $N \in \calm_{\loc}$ and the predictable process $B \in \calv$ are given by
\begin{align}\label{N-proof-2}
N &= -M + \frac{x^2}{1+x} * (\mu^M - \nu),
\\ \label{B-proof-2} B &= -A + \la M^c,M^c \ra + \frac{x^2}{1+x} * \nu.
\end{align}
Using Lemma \ref{lemma-stoch-exp-properties}, the process $S^{-1}$ is locally square-integrable if and only if $N \in \calh_{\loc}^2$. Since $M \in \calh_{\loc}^2$, this is the case if and only if
\begin{align}\label{frac-int-in-H2}
\frac{x^2}{1+x} * (\mu^M - \nu) \in \calh_{\loc}^2.
\end{align}
Since $M$ is quasi-left-continuous, by \cite[Cor. II.1.19]{Jacod-Shiryaev} there exists a version of $\nu$ that satisfies $\nu(\omega; \{ t \} \times \bbr) = 0$ for all $(\omega,t) \in \Omega \times \bbr_+$, and hence, by \cite[Thm. II.1.33.a]{Jacod-Shiryaev} we have (\ref{frac-int-in-H2}) if and only if
\begin{align*}
\bigg( \frac{x^2}{1+x} \bigg)^2 * \nu \in \cala_{\loc}^+.
\end{align*}
(ii) $\Leftrightarrow$ (iii): By \cite[Prop. II.2.29.b]{Jacod-Shiryaev} we have $x^2 * \nu \in \cala_{\loc}^+$. Therefore, the stated equivalence follows from the identity
\begin{align*}
\frac{x}{1+x} = x - \frac{x^2}{1+x} \quad \text{for all $x \in (-1,\infty)$.}
\end{align*}
(ii) $\Rightarrow$ (iv): By virtue of \cite[Prop. II.1.28]{Jacod-Shiryaev} we can define $N \in \calm_{\loc}$ as (\ref{def-N-appendix}). Using \cite[Thm. II.1.33.a]{Jacod-Shiryaev} we have $N \in \calh_{\loc}^2$, and noting the identity
\begin{align*}
- x + \frac{x^2}{1+x} = -\frac{x}{1+x} \quad \text{for all $x \in (-1,\infty)$,}
\end{align*}
we immediately see that all conditions in (\ref{N-martingale}) are fulfilled.

\noindent(iv) $\Rightarrow$ (ii): Noting (\ref{N-martingale}), we have
\begin{align*}
\Delta M + \Delta N = \Delta M - \frac{\Delta M}{1 + \Delta M} = \frac{(\Delta M)^2}{1 + \Delta M}.
\end{align*}
Since $M \in \calh_{\loc}^2$, we also have $M+N \in \calh_{\loc}^2 \cap \calv$, and hence, by \cite[Thm. I.4.56.a and b]{Jacod-Shiryaev} and Lemma \ref{lemma-jumps-compensator} we deduce (\ref{spec-both-cond}).

\noindent It remains to prove the representation (\ref{repr-S-2}). Indeed, taking into account (\ref{inverse-proof-2}), (\ref{B-proof-2}) and (\ref{N-martingale}), by Lemmas \ref{lemma-jumps-compensator} and \ref{lemma-qv-sq-qls} we obtain
\begin{align*}
S^{-1} &= \cale(N+B)
\\ &= \cale \bigg( -A + \la M^c,M^c \ra + \frac{x^2}{1 + x} * \nu + N \bigg)
\\ &= \cale \bigg( -A + \la M^c,M^c \ra + \bigg[ \sum_{s \leq \bullet} \frac{(\Delta M_s)^2}{1 + \Delta M_s} \bigg]^p + N \bigg)
\\ &= \cale \bigg( -A - \la M^c,N^c \ra - \bigg[ \sum_{s \leq \bullet} \Delta M_s \Delta N_s \bigg]^p + N \bigg)
\\ &= \cale \big( -A - \la M,N \ra + N \big),
\end{align*}
completing the proof.
\end{proof}

\section{Multiplication of stochastic exponentials}\label{sec-mult}

In this appendix we provide the required results about the multiplication of stochastic exponentials. These results are required for the analysis of the structure of an ELMD in Section \ref{sec-ELMD}. In particular, we introduce the transformations $R \mapsto \widetilde{R}$ and $\Theta \mapsto \widetilde{\Theta}$. The following auxiliary result is elementary.

\begin{lemma}\label{lemma-bijection}
The mapping $\varphi : (-1,\infty) \to (-\infty,1)$ given by
\begin{align*}
\varphi(x) = \frac{x}{1+x}, \quad x \in (-1,\infty)
\end{align*}
is bijective with inverse $\varphi^{-1} : (-\infty,1) \to (-1,\infty)$ given by
\begin{align*}
\varphi^{-1}(x) = \frac{x}{1-x}, \quad x \in (-\infty,1).
\end{align*}
\end{lemma}

Now, we introduce the transformation $R \mapsto \widetilde{R}$.

\begin{proposition}\label{prop-bij-R}
There is a bijection between the set of all predictable processes $R \in \calv$ with $\Delta R > -1$ and the set of all predictable processes $\widetilde{R} \in \calv$ with $\Delta \widetilde{R} < 1$, which is given as follows:
\begin{enumerate}
\item[(i)] For each predictable processes $R \in \calv$ with $\Delta R > -1$ we assign
\begin{align}\label{def-R-tilde}
R \mapsto \widetilde{R} := R - \sum_{s \leq \bullet} \frac{(\Delta R_s)^2}{1 + \Delta R_s}.
\end{align}
\item[(ii)] For each predictable processes $\widetilde{R} \in \calv$ with $\Delta \widetilde{R} < 1$ we assign
\begin{align}\label{def-R}
\widetilde{R} \mapsto R := \widetilde{R} + \sum_{s \leq \bullet} \frac{(\Delta \widetilde{R}_s)^2}{1 - \Delta \widetilde{R}_s}.
\end{align}
\end{enumerate}
Furthermore, for every predictable processes $R \in \calv$ with $\Delta R > -1$ and the corresponding predictable processes $\widetilde{R} \in \calv$ with $\Delta \widetilde{R} < 1$ we have
\begin{align}\label{equation-exp-R}
&\cale(-\widetilde{R}) = \cale(R)^{-1},
\\ \label{R-tilde-parts} &\widetilde{R}^c = R^c \quad \text{and} \quad \Delta \widetilde{R} = \frac{\Delta R}{1 + \Delta R},
\\ \label{R-parts} &R^c = \widetilde{R}^c \quad \text{and} \quad \Delta R = \frac{\Delta \widetilde{R}}{1 - \Delta \widetilde{R}},
\\ \label{R-quad-var} &[R,\widetilde{R}] = \sum_{s \leq \bullet} \Delta R_s \Delta \widetilde{R}_s,
\\ \label{R-R-tilde} &R = \widetilde{R} + [R,\widetilde{R}], 
\end{align}
and $R$ is continuous if and only if $\widetilde{R}$ is continuous, and in this case we have 
\begin{align}\label{exp-R-cont}
R = \widetilde{R} \quad \text{and} \quad \cale(R) = \exp(R).
\end{align}
\end{proposition}

\begin{proof}
Let $R \in \calv$ be a predictable processes with $\Delta R > -1$, and let the predictable process $\widetilde{R} \in \calv$ be given by (\ref{def-R-tilde}). Noting the equation
\begin{align*}
x - \frac{x^2}{1+x} = \frac{x}{1+x} \quad \text{for all $x \in (-1,\infty)$,}
\end{align*}
we arrive at (\ref{R-tilde-parts}). Now, let $\widetilde{R} \in \calv$ be a predictable processes with $\Delta \widetilde{R} < 1$, and let the predictable process $R \in \calv$ be given by (\ref{def-R}). Noting the equation
\begin{align*}
x + \frac{x^2}{1-x} = \frac{x}{1-x} \quad \text{for all $x \in (-\infty,1)$,}
\end{align*}
we arrive at (\ref{R-parts}). Therefore, by Lemma \ref{lemma-bijection} the mapping induced by (\ref{def-R-tilde}) and (\ref{def-R}) is a bijection. Furthermore, by \cite[Thm. I.4.52]{Jacod-Shiryaev} we have (\ref{R-quad-var}), and hence, by (\ref{def-R}) and (\ref{R-parts}) we obtain
\begin{align*}
R = \widetilde{R} + \sum_{s \leq \bullet} \frac{(\Delta \widetilde{R}_s)^2}{1 - \Delta \widetilde{R}_s} = \widetilde{R} + \sum_{s \leq \bullet} \Delta R_s \Delta \widetilde{R}_s = \widetilde{R} + [R,\widetilde{R}],
\end{align*}
showing (\ref{R-R-tilde}). Therefore, by Yor's formula (see \cite[II.8.19]{Jacod-Shiryaev}) we obtain
\begin{align*}
\cale(R) \cale(-\widetilde{R}) = \cale(R - \widetilde{R} - [R,\widetilde{R}] ) = 1,
\end{align*}
proving (\ref{equation-exp-R}). Finally, from (\ref{R-tilde-parts}) and (\ref{R-parts}) we immediately see that $\Delta R = 0$ if and only if $\Delta \widetilde{R} = 0$, and that in this case we have (\ref{exp-R-cont}).
\end{proof}

Next, we introduce the transformation $\Theta \mapsto \widetilde{\Theta}$. For this purpose, we fix a predictable process $R \in \calv$ with $\Delta R > -1$, and denote by $\widetilde{R} \in \calv$ the corresponding predictable process with $\Delta \widetilde{R} < 1$ from Proposition \ref{prop-bij-R}.

\begin{proposition}\label{prop-bij-Theta}
There is a bijection between the set of all local martingales $\Theta \in \calm_{\loc}$ with $\Theta_0 = 0$ and $\Delta \Theta < 1$ and the set of all local martingales $\widetilde{\Theta} \in \calm_{\loc}$ with $\widetilde{\Theta}_0 = 0$ and $\Delta \widetilde{\Theta} + \Delta \widetilde{R} < 1$, which is given as follows:
\begin{enumerate}
\item[(i)] For each local martingale $\Theta \in \calm_{\loc}$ with $\Theta_0 = 0$ and $\Delta \Theta < 1$ we assign
\begin{align}\label{def-Theta-tilde}
\Theta \mapsto \widetilde{\Theta} := \Theta - [\Theta,\widetilde{R}].
\end{align}
\item[(ii)] For each local martingale $\widetilde{\Theta} \in \calm_{\loc}$ with $\widetilde{\Theta}_0 = 0$ and $\Delta \widetilde{\Theta} + \Delta \widetilde{R} < 1$ we assign
\begin{align}\label{def-Theta}
\widetilde{\Theta} \mapsto \Theta := \widetilde{\Theta}^c + \frac{1}{1 - \Delta \widetilde{R}} \bdot \widetilde{\Theta}^d.
\end{align}
\end{enumerate}
Furthermore, for every local martingale $\Theta \in \calm_{\loc}$ with $\Theta_0 = 0$ and $\Delta \Theta < 1$ and the corresponding local martingales $\widetilde{\Theta} \in \calm_{\loc}$ with $\widetilde{\Theta}_0 = 0$ and $\Delta \widetilde{\Theta} + \Delta \widetilde{R} < 1$ we have
\begin{align}\label{equation-exp-R-Theta}
&\cale(-\Theta) \cale(R)^{-1} = \cale(-\widetilde{\Theta}-\widetilde{R}),
\\ \label{Theta-tilde-parts} &\widetilde{\Theta}^c = \Theta^c \quad \text{and} \quad \Delta \widetilde{\Theta} = ( 1 - \Delta \widetilde{R} ) \Delta \Theta,
\\ \label{Theta-parts} &\Theta^c = \widetilde{\Theta}^c \quad \text{and} \quad \Delta \Theta = \frac{\Delta \widetilde{\Theta}}{1 - \Delta \widetilde{R}}.
\end{align}
Furthermore, we have $\Theta = \widetilde{\Theta}$ if and only if up to an evanescent set
\begin{align}\label{R-tilde-Theta-disjoint}
\{ \Delta \widetilde{R} \neq 0 \} \cap \{ \Delta \Theta \neq 0 \} = \emptyset,
\end{align}
or equivalently, up to an evanescent set
\begin{align}\label{R-Theta-disjoint}
\{ \Delta R \neq 0 \} \cap \{ \Delta \Theta \neq 0 \} = \emptyset,
\end{align}
and $\Theta$ is quasi-left-continuous if and only if $\widetilde{\Theta}$ is quasi-left-continuous, and in this case we have $\Theta = \widetilde{\Theta}$.
\end{proposition}

\begin{proof}
Let $\Theta \in \calm_{\loc}$ be a local martingale with $\Theta_0 = 0$ and $\Delta \Theta < 1$, and let $\widetilde{\Theta}$ be the process given by (\ref{def-Theta-tilde}). By Lemma \ref{lemma-cov-R-Theta} we have $\widetilde{\Theta} \in \calm_{\loc}$. Furthermore, we have $\widetilde{\Theta}_0 = 0$ and the jumps are given by
\begin{align*}
\Delta \widetilde{\Theta} = \Delta \Theta - \Delta [ \Theta, \widetilde{R} ] = \Delta \Theta - \Delta \Theta \Delta \widetilde{R} = ( 1 - \Delta \widetilde{R} ) \Delta \Theta,
\end{align*}
showing (\ref{Theta-tilde-parts}). Furthermore, since $\Delta \widetilde{R} < 1$ and $1 - \Delta \Theta > 0$, we have
\begin{align*}
\Delta \widetilde{\Theta} + \Delta \widetilde{R} = ( 1 - \Delta \widetilde{R} ) \Delta \Theta + \Delta \widetilde{R} = \Delta \widetilde{R} (1 - \Delta \Theta) + \Delta \Theta < (1 - \Delta \Theta) + \Delta \Theta = 1.
\end{align*}
Now, let $\widetilde{\Theta} \in \calm_{\loc}$ be a local martingale with $\widetilde{\Theta}_0 = 0$ and $\Delta \widetilde{\Theta} + \Delta \widetilde{R} < 1$, and let $\Theta$ be the process given by (\ref{def-Theta}). Then we have $\Theta \in \calm_{\loc}$ with $\Theta_0 = 0$, and (\ref{Theta-parts}) is satisfied. Since $\Delta \widetilde{\Theta} < 1 - \Delta \widetilde{R}$ and $1 - \Delta \widetilde{R} > 0$, we obtain
\begin{align*}
\Delta \Theta = \frac{\Delta \widetilde{\Theta}}{1 - \Delta \widetilde{R}} < 1.
\end{align*}
Moreover, by (\ref{Theta-tilde-parts}) and (\ref{Theta-parts}) the mapping induced by (\ref{def-Theta-tilde}) and (\ref{def-Theta}) is a bijection. Using  Proposition \ref{prop-bij-R}, Yor's formula (see \cite[II.8.19]{Jacod-Shiryaev})  and (\ref{def-Theta-tilde}) we obtain
\begin{align*}
\cale(-\Theta) \cale(R)^{-1} = \cale(-\Theta) \cale(-\widetilde{R}) = \cale(-\Theta -\widetilde{R} + [\Theta,\widetilde{R}]) = \cale(-\widetilde{\Theta} - \widetilde{R}),
\end{align*}
showing (\ref{equation-exp-R-Theta}). By (\ref{Theta-tilde-parts}) we see that $\Theta = \widetilde{\Theta}$ if and only if we have (\ref{R-tilde-Theta-disjoint}) up to an evanescent set, and by Proposition \ref{prop-bij-R} this is equivalent to (\ref{R-Theta-disjoint}) up to an evanescent set. Since $1 - \Delta \widetilde{R} > 0$, by (\ref{Theta-tilde-parts}) and (\ref{Theta-parts}) we see that $\Theta$ is quasi-left-continuous if and only if $\widetilde{\Theta}$ is quasi-left-continuous, and in this case, by Lemma \ref{lemma-disjoint-jumps} we have $\Theta = \widetilde{\Theta}$.
\end{proof}

\begin{remark}
The formula (\ref{equation-exp-R-Theta}) from Proposition \ref{prop-bij-Theta} can also be obtained by using the multiplicative decomposition theorem. Indeed, the process $X = \cale(-\widetilde{R}-\widetilde{\Theta})$ has the canonical decomposition $X = 1 + M + A$ with 
\begin{align*}
M = -X_- \bdot \widetilde{\Theta} \quad \text{and} \quad A = -X_- \bdot \widetilde{R}.
\end{align*}
Therefore, by \cite[Thm. II.8.21]{Jacod-Shiryaev} we have the multiplicative decomposition $X = LD$, where the local martingale $L$ is given by
\begin{align*}
L &= \cale \bigg( \frac{1}{X_- + \Delta A} \bdot M \bigg) = \cale \bigg( - \frac{1}{X_- - X_- \Delta \widetilde{R}} \bdot (X_- \bdot \widetilde{\Theta}) \bigg)
\\ &= \cale \bigg( - \frac{1}{1 - \Delta \widetilde{R}} \bdot \widetilde{\Theta} \bigg) = \cale \bigg( - \widetilde{\Theta}^c - \frac{1}{1 - \Delta \widetilde{R}} \bdot \widetilde{\Theta}^d \bigg) = \cale(-\Theta),
\end{align*}
and where the process $D$ with locally finite variation is given by
\begin{align*}
D &= \cale \bigg( - \frac{1}{X_- + \Delta A} \bdot A \bigg)^{-1} = \cale \bigg( \frac{1}{X_- - X_- \Delta \widetilde{R}} \bdot (X_- \bdot \widetilde{R}) \bigg)^{-1}
\\ &= \cale \bigg( \frac{1}{1 - \Delta \widetilde{R}} \bdot \widetilde{R} \bigg)^{-1} = \cale \bigg( \widetilde{R}^c + \sum_{s \leq \bullet} \frac{\Delta \widetilde{R}_s}{1 - \Delta \widetilde{R}_s} \bigg)^{-1} = \cale(R)^{-1}.
\end{align*}
For the last step, we note (\ref{def-R}) and the equation
\begin{align*}
\frac{x^2}{1-x} = \frac{x}{1-x} - x \quad \text{for all $x \in (-\infty,1)$.}
\end{align*}
\end{remark}

\section{A version of Girsanov's theorem}

In this section we establish a version of Girsanov's theorem for the particular situation with a special semimartingale and an equivalent measure change. Let $\Theta \in \calm_{\loc}$ be a local martingale such that $\Theta_0 = 0$ and $\Delta \Theta < 1$. We define the local martingale $D \in \calm_{\loc}$ as the stochastic exponential $D := \cale(-\Theta)$. We assume that $D \in \calm$ with $\bbp(D_{\infty} > 0) = 1$. Let $\bbq \approx \bbp$ be the probability measure on $(\Omega,\calf_{\infty -})$ with density process $D$ relative to $\bbp$.

\begin{proposition}\label{prop-Girsanov-1}
Let $M \in \calm_{\loc}$ with $M_0 = 0$ be such that $[M,\Theta] \in \cala_{\loc}$. Then the process
\begin{align*}
M' := M + [M,\Theta]^p
\end{align*}
is a $\bbq$-local martingale, where the predictable compensator $[M,\Theta]^p$ is computed under $\bbp$.
\end{proposition}

\begin{proof}
Using \cite[Thm. I.3.18]{Jacod-Shiryaev} we have
\begin{align*}
\frac{1}{D_-} \bdot [M,D]^p = \bigg( \frac{1}{D_-} \bdot [M,D] \bigg)^p = \bigg[ M, \frac{1}{D_-} \bdot D \bigg]^p = [ M,\call(D) ]^p = -[M,\Theta]^p.
\end{align*}
Hence, the assertion follows from \cite[Thm. III.3.11]{Jacod-Shiryaev}.
\end{proof}

\begin{corollary}\label{cor-Girsanov-1}
Let $X$ be a special semimartingale with canonical decomposition $X = X_0 + M + A$ such that $[M,\Theta] \in \cala_{\loc}$. Then $X$ is also a special semimartingale under $\bbq$, and its canonical decomposition $X = X_0 + M' + A'$ is given by
\begin{align*}
M' = M + [M,\Theta]^p \quad \text{and} \quad A' = A - [M,\Theta]^p.
\end{align*}
\end{corollary}

\begin{proof}
This is an immediate consequence of Proposition \ref{prop-Girsanov-1}.
\end{proof}

\begin{lemma}\label{lemma-Girsanov-1}
Let $X$ be a special semimartingale with canonical decomposition $X = X_0 + M + A$, and let $N \in \calm_{\loc}$ be a local martingale. Then the following statements are equivalent:
\begin{enumerate}
\item[(i)] We have $[X,N] \in \cala_{\loc}$.

\item[(ii)] We have $[M,N] \in \cala_{\loc}$.
\end{enumerate}
In either case, we have $[X,N]^p = [M,N]^p$.
\end{lemma}

\begin{proof}
By \cite[Prop. I.4.49.c]{Jacod-Shiryaev} we have $[A,N] \in \calm_{\loc}$. Hence, by \cite[Lemma I.3.11]{Jacod-Shiryaev} we have $[A,N] \in \cala_{\loc}$, which proves the equivalence (i) $\Leftrightarrow$ (ii). Furthermore, by \cite[I.3.22]{Jacod-Shiryaev} we have $[A,N]^p = 0$, which concludes the proof.
\end{proof}

\begin{proposition}\label{prop-Girsanov}
Let $X$ be an $\bbr^d$-valued special semimartingale with canonical decomposition $X = X_0 + M + A$ and characteristics $(A,C,\nu)$. Then the following statements are equivalent:
\begin{enumerate}
\item[(i)] We have $[X^i,\Theta] \in \cala_{\loc}$ for all $i=1,\ldots,d$.

\item[(ii)] We have $[M^i,\Theta] \in \cala_{\loc}$ for all $i=1,\ldots,d$.
\end{enumerate}
In either case, we have $[X^i,\Theta]^p = [M^i,\Theta]^p$ for all $i=1,\ldots,d$, and the process $X$ is a special semimartingale under $\bbq$ with canonical decomposition $X = X_0 + M' + A'$ given by
\begin{align}\label{M-strich}
(M')^i &= M^i + [M^i,\Theta]^p, \quad i=1,\ldots,d,
\\ \label{A-strich} (A')^i &= A^i - [M^i,\Theta]^p, \quad i=1,\ldots,d,
\end{align}
and characteristics $(A',C',\nu')$ given by
\begin{align}\label{char-Q-1}
(A')^i &= A^i - [X^i,\Theta]^p, \quad i=1,\ldots,d,
\\ \label{char-Q-2} C' &= C,
\\ \label{char-Q-3} \nu' &= \big( 1 - M_{\mu^X}^{\bbp}( \Delta \Theta \,|\, \widetilde{\calp} ) \big) \cdot \nu.
\end{align}
\end{proposition}

\begin{proof}
The stated equivalence (i) $\Leftrightarrow$ (ii) follows from Lemma \ref{lemma-Girsanov-1}. Furthermore, by Corollary \ref{cor-Girsanov-1} the process $X$ is a special semimartingale under $\bbq$ with canonical decomposition $X = X_0 + M' + A'$ given by (\ref{M-strich}) and (\ref{A-strich}). Concerning the characteristics $(A',C',\nu')$, we immediately see that the first characteristic $A'$ is given by (\ref{char-Q-1}). According to \cite[Thm. III.3.24]{Jacod-Shiryaev} the second characteristic $C'$ is given by (\ref{char-Q-2}), and there exists a predictable nonnegative function $Y : \widetilde{\Omega} \to \bbr_+$ such that the third characteristic is given by
\begin{align*}
\nu' = Y \cdot \nu,
\end{align*}
and the function $Y$ satisfies
\begin{align*}
Y D_- = M_{\mu^X}^{\bbp}( D \,|\, \widetilde{\calp} ).
\end{align*}
Noting that
\begin{align*}
D = 1 - D_- \bdot \Theta,
\end{align*}
we have
\begin{align*}
\Delta D = - D_- \Delta \Theta.
\end{align*}
Therefore, we obtain
\begin{align*}
\frac{D}{D_-} = \frac{D_- + \Delta D}{D_-} = 1 - \Delta \Theta,
\end{align*}
and hence
\begin{align*}
Y = M_{\mu^X}^{\bbp} \bigg( \frac{D}{D_-} \,\bigg|\, \widetilde{\calp} \bigg) = M_{\mu^X}^{\bbp} ( 1 - \Delta \Theta \,|\, \widetilde{\calp} ) = 1 - M_{\mu^X}^{\bbp} ( \Delta \Theta \,|\, \widetilde{\calp} ),
\end{align*}
showing (\ref{char-Q-3}).
\end{proof}

\section{Integral characteristics of semimartingales}\label{sec-characteristics}

In this appendix we provide the results about integral characteristics of semimartingales, which we require in Section \ref{sec-existence}. We start with an auxiliary result. Let $\bbs_+^{d \times d}$ the convex cone of all symmetric, positive semidefinite $d \times d$-matrices.

\begin{lemma}\label{lemma-matrix-non-diag-elem}
Let $A \in \bbs_+^{d \times d}$ be arbitrary. Then the following statements are true:
\begin{enumerate}
\item For all $x,y \in \bbr^d$ we have
\begin{align*}
| \la Ax,y \ra_{\bbr^d} | \leq \frac{1}{2} \big( \la Ax,x \ra_{\bbr^d} + \la Ay,y \ra_{\bbr^d} \big).
\end{align*}
\item In particular, for all $i,j = 1,\ldots,d$ we have
\begin{align*}
| A^{ij} | \leq \frac{1}{2} \big( A^{ii} + A^{jj} \big).
\end{align*}
\end{enumerate}
\end{lemma}

\begin{proof}
For all $x,y \in \bbr^d$ we have by polarization
\begin{align*}
| \la Ax,y \ra_{\bbr^d} | &= \frac{1}{4} | \la A(x+y), x+y \ra_{\bbr^d} - \la A(x-y), x-y \ra_{\bbr^d} |
\\ &\leq \frac{1}{4} \big( \la A(x+y), x+y \ra_{\bbr^d} + \la A(x-y), x-y \ra_{\bbr^d} \big)
\\ &= \frac{1}{2} \big( \la Ax,x \ra_{\bbr^d} + \la Ay,y \ra_{\bbr^d} \big),
\end{align*}
proving the first statement. The second statement is an immediate consequence by taking $x = e_i$ and $y = e_j$ for all $i,j = 1,\ldots,d$.
\end{proof}

Now, let $X$ be an $\bbr^d$-valued locally square-integrable, quasi-left-continuous semimartingale with canonical decomposition 
\begin{align*}
X = X_0 + M + A. 
\end{align*}

\begin{definition}
We introduce the following notions:
\begin{enumerate}
\item Let $C \in \calv^{d \times d}$ be the continuous $\bbs_+^{d \times d}$-valued process given by
\begin{align*}
C^{ij} = \la M^{i,c},M^{j,c} \ra, \quad i,j=1,\ldots,d.
\end{align*}
\item Let $C_{\mo} \in \calv^{d \times d}$ be the continuous $\bbs_+^{d \times d}$-valued process given by
\begin{align*}
C_{\mo}^{ij} = \la M^{i},M^{j} \ra, \quad i,j=1,\ldots,d.
\end{align*}
\item Let $V \in \calv^{d \times d}$ be the continuous $\bbs_+^{d \times d}$-valued process given by
\begin{align*}
V^{ij} = \la M^{i,d},M^{j,d} \ra, \quad i,j=1,\ldots,d.
\end{align*}

\item Let $\nu$ is the predictable compensator of the random measure $\mu^X$ associated to the jumps of $X$.

\item We call the triplet $(A,C,\nu)$ the \emph{characteristics} of $X$.

\item We call $C_{\mo}$ the \emph{modified second characteristic} of $X$.

\item We call the triplet $(A,C_{\mo},\nu)$ the \emph{modified characteristics} of $X$.

\item We call $V$ the \emph{purely discontinuous second characteristic} of $X$.
\end{enumerate}
\end{definition}

\begin{remark}
According to \cite[II.2.12]{Jacod-Shiryaev} we may assume that for all $0 \leq s \leq t$ we have
\begin{align*}
C_t - C_s, V_t - V_s, C_{\mo,t} - C_{\mo,s} \in \bbs_+^{d \times d}.
\end{align*}
\end{remark}

\begin{lemma}\label{lemma-sum-char}
We have $C_{\mo} = C + V$.
\end{lemma}

\begin{proof}
This is clear, because $\la M,N \ra = 0$ for two local martingales $M,N \in \calh_{\loc}^2$ such that $M$ is continuous and $N$ is purely discontinuous.
\end{proof}

\begin{lemma}\label{lemma-pd-char}
We have $V^{ij} = (x^i x^j) * \nu$ for all $i,j=1,\ldots,d$.
\end{lemma}

\begin{proof}
Since $X$ is quasi-left-continuous, this is a consequence of Lemma \ref{lemma-sum-char} and \cite[Prop. II.2.17]{Jacod-Shiryaev}.
\end{proof}

\begin{lemma}\label{lemma-obtain-S-valued-proc}
Let $C$ be an optional $\bbs_+^{d \times d}$-valued process. Furthermore, let $\Gamma \in \calv^+$ be such that $d C^{ii} \ll d \Gamma$ for all $i=1,\ldots,d$. Then there exists a $\bbs_+^{d \times d}$-valued optional process $c$ such that $C = c \bdot \Gamma$.
\end{lemma}

\begin{proof}
We have $d C^{ij} \ll d \Gamma$ for all $i,j = 1,\ldots,d$. Indeed, let $0 \leq s \leq t$ be such that $\Gamma_t - \Gamma_s = 0$. Then we have $C_t^{ii} - C_s^{ii} = 0$ for all $i=1,\ldots,d$. Since $C_t - C_s \in \bbs_+^{d \times d}$, by Lemma \ref{lemma-matrix-non-diag-elem} it follows that $C_t^{ij} - C_s^{ij} = 0$ for all $i,j = 1,\ldots,d$. Therefore, by \cite[Prop. I.3.13]{Jacod-Shiryaev} there exists an $\bbr^{d \times d}$-valued, optional process $c \in L_{\loc}^1(\Gamma)$ such that $C = c \bdot \Gamma$. Proceeding as in step (c) of the proof of \cite[Prop. II.2.9]{Jacod-Shiryaev} we obtain, after changing $c$ on an evanescent set if necessary, that $c$ is $\bbs_+^{d \times d}$-valued.
\end{proof}

\begin{definition}
Let $\Gamma \in \cala_{\loc}^+$ be a continuous process. We call a triplet $(a,c,K)$ \emph{integral characteristics} of $X$ with respect to $\Gamma$ if the following conditions are fulfilled:
\begin{enumerate}
\item $a$ is an optional $\bbr^d$-valued process such that $a^i \in L_{\loc}^1(\Gamma)$ for all $i=1,\ldots,d$.

\item $c$ is an optional $\bbs_+^{d \times d}$-valued process such that $c^{ii} \in L_{\loc}^1(\Gamma)$ for all $i=1,\ldots,d$.

\item $K$ is a transition kernel from $(\Omega \times \bbr_+, \calo)$ into $(\bbr^d,\calb(\bbr^d))$ such that on $\Omega \times \bbr_+$ we have
\begin{align*}
K(\{ 0 \}) = 0 \quad \text{and} \quad \int_{\bbr^d} |x|^2 K(dx) < \infty.
\end{align*}

\item We have $B = b \bdot \Gamma$, $C = c \bdot \Gamma$ and $\nu = K \otimes \Gamma$.

\end{enumerate}
\end{definition}

\begin{remark}
Since $\Gamma$ is continuous, it suffices that $a$, $c$ and $K$ are optional rather than predictable.
\end{remark}

\begin{remark}
Let $c$ be an optional $\bbs_+^{d \times d}$-valued process such that $c^{ii} \in L_{\loc}^1(\Gamma)$ for all $i=1,\ldots,d$. Then, by Lemma \ref{lemma-matrix-non-diag-elem} we also have $c^{ij} \in L_{\loc}^1(\Gamma)$ for all $i,j=1,\ldots,d$.
\end{remark}

\begin{definition}
Let $\Gamma \in \cala_{\loc}^+$ be a continuous process. We call a process $c_{\mo}$ \emph{modified second integral characteristic} of $X$ with respect to $\Gamma$ if the following conditions are fulfilled:
\begin{enumerate}
\item $c_{\mo}$ is an optional $\bbs_+^{d \times d}$-valued process such that $c_{\mo}^{ii} \in L_{\loc}^1(\Gamma)$ for all $i=1,\ldots,d$.

\item We have $C_{\mo} = c_{\mo} \bdot \Gamma$.
\end{enumerate}
In this case, we call the triplet $(a,c_{\mo},K)$ \emph{modified integral characteristics} of $X$ with respect to $\Gamma$.
\end{definition}

\begin{definition}
Let $\Gamma \in \cala_{\loc}^+$ be a continuous process. We call a process $v$ \emph{purely discontinuous second integral characteristic} of $X$ with respect to $\Gamma$ if the following conditions are fulfilled:
\begin{enumerate}
\item $v$ is an optional $\bbs_+^{d \times d}$-valued process such that $v^{ii} \in L_{\loc}^1(\Gamma)$ for all $i=1,\ldots,d$.

\item We have $V = v \bdot \Gamma$.
\end{enumerate}
\end{definition}

\begin{proposition}\label{prop-ex-int-char}
There exist a continuous process $\Gamma \in \cala_{\loc}^+$ and modified integral characteristics $(a,c_{\mo},K)$ of $X$ with respect to $\Gamma$.
\end{proposition}

\begin{proof}
The proof is analogous to that of \cite[Prop. II.2.9]{Jacod-Shiryaev}.
\end{proof}

\begin{proposition}\label{prop-int-char-eq}
Let $\Gamma \in \cala_{\loc}^+$ be a continuous process. Then the following statements are equivalent:
\begin{enumerate}
\item[(i)] There exist modified integral characteristics $(a,c_{\mo},K)$ of $X$ with respect to $\Gamma$.

\item[(ii)] There exist integral characteristics $(a,c,K)$ and a purely discontinuous second integral characteristic $v$ of $X$ with respect to $\Gamma$.
\end{enumerate}
If the previous conditions are fulfilled, then we have
\begin{align}
c_{\mo} = c + v \quad \text{$\Gamma$-a.e.} \quad \text{$\bbp$-a.e.}
\end{align}
and for all $i,j = 1,\ldots,d$ we have
\begin{align}\label{v-repr-K}
v^{ij} = \int_{\bbr^d} x^i x^j K(dx) \quad \text{$\Gamma$-a.e.} \quad \text{$\bbp$-a.e.}
\end{align}
\end{proposition}

\begin{proof}
(i) $\Rightarrow$ (ii): By \cite[Prop. I.3.5]{Jacod-Shiryaev} we have $d C_{\mo}^{ii} \ll d \Gamma$ for all $i = 1,\ldots,d$. By Lemma \ref{lemma-sum-char} we have $C_{\mo} = C + V$, and hence it follows that $d C^{ii} \ll d \Gamma$ and $V^{ii} \ll d \Gamma$ for all $i=1,\ldots,d$. Hence, by Lemma \ref{lemma-obtain-S-valued-proc} there exist optional $\bbs_+^{d \times d}$-valued processes $c,v \in L_{\loc}^1(\Gamma)$ such that $C = c \bdot \Gamma$ and $V = v \bdot \Gamma$.

\noindent (ii) $\Rightarrow$ (i): We define the optional process $c_{\mo} := c + v$. Then we have $c_{\mo} \in L_{\loc}^1(\Gamma)$, and by Lemma \ref{lemma-sum-char} we obtain
\begin{align*}
C_{\mo} = C + V = c \bdot \Gamma + v \bdot \Gamma = (c+v) \bdot \Gamma = c_{\mo} \bdot \Gamma,
\end{align*}
completing the proof of this implication. The additional statement (\ref{v-repr-K}) follows from Lemma \ref{lemma-pd-char}.
\end{proof}

\section{Matrices and linear operators}\label{sec-matrices}

In this appendix we provide the required results about matrices and linear operators. We denote by $\bbs^{d \times d} \subset \bbr^{d \times d}$ the subspace of all symmetric, real-valued matrices. Furthermore, we denote by $\bbs_+^{d \times d} \subset \bbs^{d \times d}$ the convex cone of all symmetric, positive semidefinite matrices, and we denote by $\bbs_{++}^{d \times d} \subset \bbs_+^{d \times d}$ the subset of all symmetric, positive definite matrices. For a matrix $A \in \bbr^{d \times m}$ we denote by $A^{\dagger} \in \bbr^{m \times d}$ the Moore-Penrose inverse; see \cite[page 649]{Boyd}.

\begin{lemma}\label{lemma-MP-measurable}
The Moore-Penrose inverse $\bbr^{d \times m} \to \bbr^{m \times d}$, $A \mapsto A^{\dagger}$ is measurable.
\end{lemma}

\begin{proof}
This follows from the representation
\begin{align*}
A^{\dagger} = \lim_{\epsilon \to 0} ( A^{\top} A + \epsilon \Id )^{-1} A^{\top},
\end{align*}
see \cite[page 649]{Boyd}.
\end{proof}

\begin{lemma}\label{lemma-matrix-pre}
Let $\hat{A} \in \bbs^{(d+1) \times (d+1)}$ be a symmetric matrix of the form
\begin{align}\label{matrix-A}
\hat{A} = \left(
\begin{array}{cc}
A & b
\\ b^{\top} & c
\end{array}
\right)
\end{align}
with a symmetric, positive semidefinite matrix $A \in \bbs_+^{d \times d}$, a vector $b \in \bbr^d$ and a real number $c \in \bbr$. Then the following statements are equivalent:
\begin{enumerate}
\item[(i)] $\hat{A}$ is positive semidefinite; that is $\hat{A} \in \bbs_+^{(d+1) \times (d+1)}$.

\item[(ii)] We have $(\Id - A A^{\dagger})b = 0$ and the Schur complement satisfies 
\begin{align*}
c - \la A^{\dagger} b,b \ra_{\bbr^d} \geq 0.
\end{align*}

\end{enumerate}
\end{lemma}

\begin{proof}
See \cite[page 651]{Boyd}.
\end{proof}

\begin{lemma}\label{lemma-LGS}
For a matrix $A \in \bbr^{d \times m}$ and a vector $b \in \bbr^d$ the following statements are equivalent:
\begin{enumerate}
\item[(i)] The system of linear equations
\begin{align}\label{LGS}
Ax = b, \quad x \in \bbr^m 
\end{align}
has a solution.

\item[(ii)] We have $A A^{\dagger} b = b$; that is $(\Id - A A^{\dagger})b = 0$.
\end{enumerate}
In either case, a solution to (\ref{LGS}) is given by $x = A^{\dagger} b$.
\end{lemma}

\begin{proof}
(i) $\Rightarrow$ (ii): According to \cite[page 649]{Boyd} have $A A^{\dagger} = \Pi_{\ran(A)}$, the orthogonal projection on the range of $A$. Since $b \in \ran(A)$, we obtain
\begin{align*}
A A^{\dagger} b = \Pi_{\ran(A)} b = b.
\end{align*}

\noindent(ii) $\Rightarrow$ (i): By hypothesis, a solution to (\ref{LGS}) is given by $x = A^{\dagger} b$.
\end{proof}

\begin{proposition}\label{prop-matrix}
Let $A \in \bbs_+^{d \times d}$ be a symmetric, positive semidefinite matrix, and let $b \in \bbr^d$ be a vector. Then the following statements are equivalent:
\begin{enumerate}
\item[(i)] There exists a symmetric, positive semidefinite matrix $\hat{A} \in \bbs_+^{(d+1) \times (d+1)}$ such that $\hat{A}^{ij} = A^{ij}$ for all $i,j = 1,\ldots,d$ and $\hat{A}^{i,d+1} = b^i$ for all $i=1,\ldots,d$.

\item[(ii)] The system of linear equations
\begin{align*}
Ax = b, \quad x \in \bbr^d 
\end{align*}
has a solution.
\end{enumerate}
If the previous conditions are fulfilled, then for every $c \geq \la A^{\dagger} b,b \ra_{\bbr^d}$ the symmetric matrix (\ref{matrix-A}) is positive semidefinite.
\end{proposition}

\begin{proof}
This is a consequence of Lemmas \ref{lemma-matrix-pre} and \ref{lemma-LGS}.
\end{proof}

\begin{lemma}\label{lemma-LGS-2}
Let $A \in \bbr^{d \times m}$ and $B \in \bbr^{d \times n}$ be matrices, and let $b \in \bbr^d$ be a vector. Then the following statements are equivalent:
\begin{enumerate}
\item[(i)] The system of linear equations
\begin{align}\label{LGS-2}
Ax + By = b, \quad \text{$x \in \bbr^m$ and $y \in \bbr^n$} 
\end{align}
has a solution.

\item[(ii)] There exists $c \in \bbr^d$ such that the system of linear equations
\begin{align*}
\left\{
\begin{array}{ll}
Ax = c, & x \in \bbr^m
\\ By = b-c, & y \in \bbr^n
\end{array}
\right.
\end{align*}
has a solution.

\item[(iii)] There exists $c \in \bbr^d$ such that
\begin{align*}
(\Id - A A^{\dagger}) c = 0 \quad \text{and} \quad (\Id - B B^{\dagger}) (b-c) = 0.
\end{align*}

\end{enumerate}
\end{lemma}

\begin{proof}
(i) $\Leftrightarrow$ (ii): This equivalence is obvious.

\noindent(ii) $\Leftrightarrow$ (iii): This equivalence follows from Lemma \ref{lemma-LGS}.
\end{proof}

\begin{proposition}\label{prop-matrix-2}
Let $A,B \in \bbs_+^{d \times d}$ be symmetric, positive semidefinite matrices, and let $b \in \bbr^d$ be a vector. Then the following statements are equivalent:
\begin{enumerate}
\item[(i)] There exist symmetric, positive semidefinite matrices $\hat{A},\hat{B} \in \bbs_+^{(d+1) \times (d+1)}$ such that $\hat{A}^{ij} = A^{ij}$ and $\hat{B}^{ij} = B^{ij}$ for all $i,j = 1,\ldots,d$ as well as $\hat{A}^{i,d+1} + \hat{B}^{i,d+1} = b^i$ for all $i=1,\ldots,d$.

\item[(ii)] The system of linear equations
\begin{align}\label{LGS-2-Rd}
Ax + By = b, \quad x,y \in \bbr^d 
\end{align}
has a solution.
\end{enumerate}
\end{proposition}

\begin{proof}
The first statement is satisfied if and only if there exist $c \in \bbr^d$ and symmetric, positive semidefinite matrices $\hat{A},\hat{B} \in \bbs_+^{(d+1) \times (d+1)}$ such that $\hat{A}^{ij} = A^{ij}$ and $\hat{B}^{ij} = B^{ij}$ for all $i,j = 1,\ldots,d$ as well as $\hat{A}^{i,d+1} = c^i$ and $\hat{B}^{i,d+1} = b^i - c^i$ for all $i=1,\ldots,d$. By Proposition \ref{prop-matrix} this is the case if and only if there exists $c \in \bbr^d$ such that the system of linear equations
\begin{align*}
\left\{
\begin{array}{ll}
Ax = c, & x \in \bbr^d
\\ By = b-c, & y \in \bbr^d
\end{array}
\right.
\end{align*}
has a solution. According to Lemma \ref{lemma-LGS-2} this is equivalent to the existence of a solution to the system of linear equations (\ref{LGS-2-Rd}).
\end{proof}

\begin{proposition}\label{prop-matrices-main}
Let $A,B \in \bbs_+^{d \times d}$ be arbitrary, and set $C := A+B \in \bbs_+^{d \times d}$. Furthermore, let $b \in \bbr^d$ be arbitrary. Then the following statements are equivalent:
\begin{enumerate}
\item[(i)] There exists a symmetric, positive semidefinite matrix $\hat{C} \in \bbs_+^{(d+1) \times (d+1)}$ such that $\hat{C}^{ij} = C^{ij}$ for all $i,j = 1,\ldots,d$ and $\hat{C}^{i,d+1} = b^i$ for all $i=1,\ldots,d$.

\item[(ii)] There exist symmetric, positive semidefinite matrices $\hat{A}, \hat{B} \in \bbs_+^{(d+1) \times (d+1)}$ such that $\hat{A}^{ij} = A^{ij}$ and $\hat{B}^{ij} = B^{ij}$ for all $i,j = 1,\ldots,d$ as well as $\hat{A}^{i,d+1} + \hat{B}^{i,d+1} = b^i$ for all $i=1,\ldots,d$.

\item[(iii)] The system of linear equations
\begin{align*}
Cx = b, \quad x \in \bbr^d
\end{align*}
has a solution.

\item[(iv)] The system of linear equations
\begin{align*}
Ax + By = b, \quad x,y \in \bbr^d
\end{align*}
has a solution.

\end{enumerate}
\end{proposition}

\begin{proof}
(i) $\Leftrightarrow$ (iii): This equivalence follows from Proposition \ref{prop-matrix}.

\noindent(ii) $\Leftrightarrow$ (iv): This equivalence follows from Proposition \ref{prop-matrix-2}.

\noindent(iii) $\Rightarrow$ (iv): Since $C = A+B$, this implication is obvious.

\noindent(ii) $\Rightarrow$ (i): The matrix $\hat{C} := \hat{A} + \hat{B}$ has the desired properties.
\end{proof}

\begin{lemma}\label{lemma-measure}
Let $K$ be a measure on $(\bbr^d,\calb(\bbr^d))$ such that 
\begin{align*}
K(\{ 0 \}) = 0 \quad \text{and} \quad \int_{\bbr^d} |x|^2 K(dx) < \infty. 
\end{align*}
Furthermore, let $\hat{A} \in \bbs_+^{(d+1) \times (d+1)}$ be a symmetric, positive semidefinite matrix of the form (\ref{matrix-A}) such that
\begin{align*}
A^{ij} = \int_{\bbr^d} x^i x^j K(dx) \quad \text{for all $i,j = 1,\ldots,d$.}
\end{align*}
Then there exists a measure $\hat{K}$ on $(\bbr^{d+1},\calb(\bbr^{d+1}))$ with 
\begin{align}\label{measure-matrix-1}
\hat{K}(\{ 0 \}) = 0 \quad \text{and} \quad \int_{\bbr^{d+1}} |\hat{x}|^2 \hat{K}(d \hat{x}) < \infty
\end{align}
such that for every nonnegative, measurable function $\hat{f} : \bbr^{d+1} \to \bbr_+$ we have
\begin{align}\label{measure-matrix-2}
\int_{\bbr^{d+1}} \hat{f}(\hat{x}) \hat{K}(d \hat{x}) = \int_{\bbr^d} \hat{f}(x,\la A^{\dagger} b,x \ra_{\bbr^d}) K(dx),
\end{align}
for every nonnegative, measurable function $f : \bbr^d \to \bbr_+$ we have
\begin{align}\label{measure-matrix-3}
\int_{\bbr^d} f(x) K(dx) = \int_{\bbr^{d+1}} f(x) \hat{K}(d \hat{x}),
\end{align}
and for all $i,j = 1,\ldots,d+1$ such that $i \leq d$ or $j \leq d$ we have
\begin{align}\label{measure-matrix-4}
\hat{A}^{ij} = \int_{\bbr^{d+1}} \hat{x}^i \hat{x}^j \hat{K}(d\hat{x}).
\end{align}
\end{lemma}

\begin{proof}
By Proposition \ref{prop-matrix} there exists a solution $y \in \bbr^d$ to the system of linear equations
\begin{align*}
Ay = b, \quad y \in \bbr^d,
\end{align*}
and according to Lemma \ref{lemma-LGS} one such solution is given by $y = A^{\dagger} b$. We define the linear mapping $\ell : \bbr^d \to \bbr^{d+1}$ as
\begin{align*}
\ell(x) := (x,\la y,x \ra_{\bbr^d}), \quad x \in \bbr^d
\end{align*}
and the image measure $\hat{K} := K \circ \ell$. Since $\ell$ is one-to-one, we obtain
\begin{align*}
\hat{K}(\{ 0 \}) = K(\ell^{-1}(\{ 0 \})) = K(\{ 0 \}) = 0.
\end{align*}
Moreover, we have
\begin{align*}
\int_{\bbr^{d+1}} |\hat{x}|^2 \hat{K}(d \hat{x}) = \int_{\bbr^d} |\ell(x)|^2 K(dx) = \int_{\bbr^d} \big( |x|^2 + |\la y,x \ra_{\bbr^d}|^2 \big) K(dx) < \infty,
\end{align*}
showing (\ref{measure-matrix-1}). Furthermore, we have (\ref{measure-matrix-2}). Let $f : \bbr^d \to \bbr_+$ be a nonnegative, measurable function, and let $\hat{f} : \bbr^{d+1} \to \bbr_+$ be the extension given by
\begin{align*}
\hat{f}(\hat{x}) = f(x) \quad \text{for each $\hat{x} = (x,y) \in \bbr^{d+1}$.}
\end{align*}
Then by (\ref{measure-matrix-2}) we have
\begin{align*}
\int_{\bbr^d} f(x) K(dx) = \int_{\bbr^d} \hat{f}(x,\la y,x \ra_{\bbr^d}) K(dx) = \int_{\bbr^{d+1}} \hat{f}(\hat{x}) \hat{K}(d \hat{x}) = \int_{\bbr^{d+1}} f(x) \hat{K}(d \hat{x}),
\end{align*}
showing (\ref{measure-matrix-3}). Furthermore, since $Ay = b$, for each $i=1,\ldots,d$ we have
\begin{align*}
\int_{\bbr^{d+1}} \hat{x}^i \hat{x}^{d+1} \hat{K}(d\hat{x}) &= \int_{\bbr^d} x^i \la y,x \ra_{\bbr^d} K(dx) = \sum_{j=1}^d y^j \int_{\bbr^d} x^i x^j K(dx) 
\\ &= \sum_{j=1}^d A^{ij} y^j = b^i = \hat{A}^{i,d+1},
\end{align*}
proving (\ref{measure-matrix-4}).
\end{proof}

\begin{lemma}\label{lemma-solution-kernel}
Let $X$ be a Hilbert space. Then for every continuous linear operator $T \in L(X,\bbr^d)$ and every $y \in \bbr^d$ following statements are equivalent:
\begin{enumerate}
\item[(i)] There exists $x \in X$ such that
\begin{align*}
Tx = y.
\end{align*}
\item[(ii)] There exists $\eta \in \bbr^d$ such that
\begin{align*}
T T^* \eta = y.
\end{align*}
\end{enumerate}
\end{lemma}

\begin{proof}
(i) $\Rightarrow$ (ii): The range $\ran(T^*)$ is a finite dimensional subspace of $X$, and hence it is closed. Therefore, we have the direct sum decomposition
\begin{align*}
X = \ran(T^*) \oplus \ker(T).
\end{align*}
Let $x = x_1 + x_2$ be the corresponding decomposition of $x$. Then we have $T x_1 = y$. Since $x_1 \in \ran(T^*)$, there exists $\eta \in \bbr^d$ such that $T^* \eta = x_1$, and hence $T T^* \eta = y$.

\noindent(ii) $\Rightarrow$ (i): Taking $x = T^* \eta$ we have $Tx = y$.
\end{proof}

\begin{lemma}
Let $X$ and $Y$ be Hilbert spaces, and let $T \in L(X,\bbr^d)$ and $S \in L(Y,\bbr^d)$ be continuous linear operators. We define $R \in L(X \oplus_2 Y,\bbr^d)$ as
\begin{align*}
R(x,y) := Tx + Sy, \quad x \in X \text{ and } y \in Y.
\end{align*}
Then the following statements are true:
\begin{enumerate}
\item We have $R^* = (T^*,S^*)$.

\item We have $R R^* = T T^* + S S^*$.
\end{enumerate}
\end{lemma}

\begin{proof}
For all $x \in X$, $y \in Y$ and $z \in \bbr^d$ we have
\begin{align*}
\la R(x,y),z \ra_{\bbr^d} &= \la Tx + Sy, z \ra_{\bbr^d} = \la x,T^* z \ra_X + \la y,S^* z \ra_Y 
\\ &= \la (x,y), (T^*z,S^*z) \ra_{X \oplus_2 Y}, 
\end{align*}
proving the first statement. Now, the second statement is an immediate consequence.
\end{proof}

\begin{lemma}\label{lemma-adjungiert}
Let $X$ be a Hilbert space, and let $T \in L(X,\bbr^d)$ be a continuous linear operator. Then we have
\begin{align*}
T T^* y = A \cdot y, \quad y \in \bbr^d,
\end{align*}
where $A \in \bbr^{d \times d}$ is the matrix given by
\begin{align*}
A_{ij} = \la x_i,x_j \ra, \quad i,j = 1,\ldots,d,
\end{align*}
and where $x_1,\ldots,x_d \in X$ are the unique elements such that
\begin{align*}
T = ( \la x_1,\cdot \ra, \ldots, \la x_d,\cdot \ra ).
\end{align*}
\end{lemma}

\begin{proof}
For each $i=1,\ldots,d$ we have
\begin{align*}
\la T T^* y,e_i \ra_{\bbr^d} = \la x_i,T^* y \ra = \la Tx_i,y \ra_{\bbr^d} = \sum_{j=1}^d \la x_i,x_j \ra y_j = \la A \cdot y,e_i \ra_{\bbr^d},
\end{align*}
completing the proof.
\end{proof}

\end{appendix}

\bibliographystyle{plain}

\bibliography{Finance}

\end{document}